\documentclass{llncs}

\usepackage[utf8]{inputenc}
\usepackage{amssymb,amsmath,amsfonts}
\usepackage[all]{xy}
\SelectTips{cm}{}
\usepackage{fancybox}
\usepackage{color}
\usepackage{graphicx}
\usepackage{url}
\usepackage{xspace}
\usepackage{enumitem} 
\usepackage{wrapfig} 

\newcommand{\To}{\Rightarrow}
\newcommand{\rto}{\rightarrow}
\newcommand{\uprto}[1]{\stackrel{#1}{\rto}}
\newcommand{\lto}{\leftarrow}
\newcommand{\uplto}[1]{\stackrel{#1}{\lto}}
\newcommand{\monoto}{\rightarrowtail}
\newcommand{\upmonoto}[1]{\stackrel{#1}{\monoto}}

\newcommand{\parto}{\rightharpoonup} 
\newcommand{\slice}{\downarrow}

\newcommand{\Set}{\ensuremath{\mathbf{Set}}}

\newcommand{\catC}{\ensuremath{\mathbf{C}}}
\newcommand{\catP}{\ensuremath{\mathbf{P}_\mathbf{C}^\mathcal{M}}}

\newcommand{\id}{\mathit{id}}
\newcommand{\Id}{\mathit{Id}}
\newcommand{\ol}{\overline}

\newcommand{\olj}{\overline{\jmath}}
\newcommand{\olt}{\overline{t}}
\newcommand{\olm}{\overline{m}}
\newcommand{\oln}{\overline{n}}

\newcommand{\pb}{PB\xspace}  
\newcommand{\fpbc}{FPBC\xspace}  
  
\newcommand{\sqpo}{SqPO\xspace}  
\newcommand{\pbcpo}{{\sc AGREE}\xspace}  
\newcommand{\agree}{\pbcpo} 

\newcommand{\hide}[1]{}

\newcommand{\grpol}[1]{\mathbb{#1}}   
\newcommand{\Grpol}{\mathbf{Gr^\pm}} 
\newcommand{\depol}{\mathrm{Depol}} 
\newcommand{\Gr}{\mathbf{Gr}} 
\newcommand{\polar}{\mathrm{Pol}}
\newcommand{\spa}[5]{#1\stackrel{#2}{\leftarrow}#3\stackrel{#4}{\rightarrow}#5}
\newcommand{\psqpo}{PSqPO\xspace}
 
\newcommand{\step}[3]{\xymatrix{#1 \ar@{=>}[r]^{\mathrm{#2}} & #3 \\ } }

\newcommand{\stepp}[4]{\xymatrix{#1 \ar@{=>}[r]^{\mathrm{#2}}_{#3} & #4 \\ } }

\newcommand{\ourbot}{*}
\newtheorem{fact}{Fact}
\newcommand{\init}{0}
\newcommand{\embspec}{embedding}
\newdir{ >}{*{}!/-.2em/\dir{>}}
\newcommand{\bx}[1]{\phantom{\big(}#1{\phantom{\big)}}}
\newcommand{\M}{\ensuremath{\mathcal{M}}}
\newcommand{\Str}{\ensuremath{\mathcal{S}}}
\newcommand{\dotarrow}{
   \mathrel{\ooalign{\hss\raise1ex\hbox{\scalebox{1.25}{\normalfont .}}%
   \kern0.35ex\hss\cr$\rightarrow$}}}
   
\newcommand{\eqid}[1]{\stackrel{\langle#1\rangle}{=}}
\newcommand{\ideq}[1]{\ensuremath{\langle#1\rangle}}
\newcommand{\sm}{\setminus}

\newcommand{\unit}{u}
\newcommand{\counit}{c}
\newcommand{\wh}{\widehat}  
\newcommand{\TT}{\mathbb{T}}
\newcommand{\ini}{0} 
\newcommand{\fin}{1} 
\newcommand{\final}{1} 
\newcommand{\inifin}{!} 
\newcommand{\ttrue}{\mathit{true}} 
\newcommand{\ffalse}{\mathit{false}} 
\newcommand{\olinifin}{\ol{\,!\,}}
\newcommand{\comment}[1]{}
\begin{document}


\title{AGREE -- Algebraic Graph Rewriting \\ with Controlled
  Embedding\thanks{This work has been partly funded by projects CLIMT
    (ANR/(ANR-11-BS02-016), TGV (CNRS-INRIA-FAPERGS/(156779 and
    12/0997-7)), VeriTeS (CNPq 485048/2012-4 and 309981/2014-0), PEPS
    \'egalit\'e (CNRS).}\\
  {\small (Long Version)}}
\author{A. Corradini\inst{1} and D. Duval\inst{2} and R. Echahed\inst{3} 
and F. Prost\inst{3} and L. Ribeiro\inst{4}}
\institute{
Dipartimento di Informatica, Universit\`a di Pisa
\and LJK - Universit\'e de Grenoble Alpes and CNRS
\and LIG - Universit\'e de Grenoble Alpes and CNRS
\and INF - Universidade Federal do Rio Grande do Sul}

\maketitle

\begin{abstract}
  The several algebraic approaches to graph transformation proposed in
  the literature all ensure that if an item is preserved by a rule, so
  are its connections with the context graph where it is embedded. But
  there are applications in which it is desirable to specify different
  embeddings. For example when cloning an item, there may be a need to
  handle the original and the copy in different ways. We propose a
  conservative extension of classical algebraic approaches to graph
  transformation, for the case of monic matches, where rules allow one
  to specify how the embedding of preserved items should be carried
  out.

\end{abstract}

\section{Introduction} 
\label{sec:introduction}


Graphs are used to describe a wide range of situations in a precise
yet intuitive way.  Different kinds
of graphs are used in modelling techniques depending on the
investigated fields, which include computer
science, chemistry, biology, quantum computing, etc.   When system states are represented by graphs, it is
natural to use rules that transform graphs to describe the system evolution.
%
There are 
two main streams in the research on 
graph transformations: (i) the algorithmic approaches, 
which  describe explicitly, with a concrete algorithm, the result of applying a rule to
a graph  
(see e.g.~\cite{EngelfrietR97,Echahed08b}), and (ii) the algebraic approaches which
define abstractly a graph transformation step using 
basic constructs borrowed from category theory. In this paper we will consider  
the latter.


The basic idea of all approaches is the same: states are represented
by graphs and state changes are represented by rules that modify
graphs.  The differences are the kind of graphs that may be used, and
the definitions of when and how rules may be applied. One critical
point when defining graph transformation is that 
one cannot delete or copy part 
of a graph without considering the effect of the
operation on the rest of the graph, because deleted/copied items may be linked to others.  
For example,
rule $\rho1$ in Figure~\ref{fig_rulesAndGraphs}(a) specifies that a node
shall be deleted and rule $\rho2$ that a node shall be
duplicated (\textsf{C} labels the copy). 
What should
be the result of applying these rules to the grey node of graph $G$ in Figure~\ref{fig_rulesAndGraphs}(b)?
Different approaches give different answers to this
question.
\begin{figure}[htbp]
\vspace{-0.5cm}
\begin{center}
\includegraphics[width=0.9\textwidth]{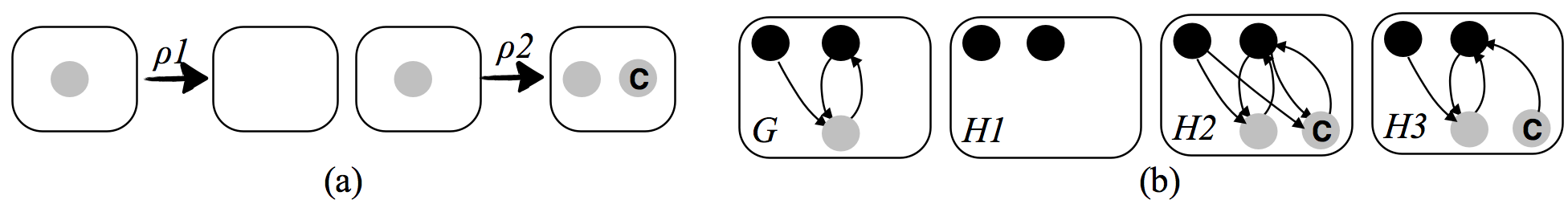} \vspace{-0.5cm}
\end{center}
\caption{(a) Delete/Copy Rules (b) Resulting Graphs }
\label{fig_rulesAndGraphs}
\end{figure}
\vspace{-0.5cm}

The most popular algebraic approaches are the double-pushout (DPO) and the single-pushout 
(SPO), which can be
illustrated as follows:
\vspace{-0.3cm}
$$ \begin{array}{ccc}
  \xymatrix@C=4pc@R=1.5pc{
 \ar@{}[rd]|{PO}  L \ar[d]_{m} & \ar@{}[rd]|{PO}K \ar[l]|{\bx{l}} \ar[d]^{d}  \ar[r]|{\bx{r}} 
  & R \ar[d]^{m'} \\ 
  G 
  & D \ar[l]|{\bx{l'}} \ar[r]|{\bx{r'}} 
  & H 
  \\ 
  } &  \hspace{5mm} & 
  \xymatrix@C=6pc@R=1.5pc{
  \ar@{}[rd]|{PO} L \ar[d]_{m} \ar[r]|{\bx{\psi}}& R  \ar[d]^{m'}  \\ 
  G  \ar[r]|{\bx{\psi'}} & H 
  \\ 
  } \\
  \mbox{Double pushout rewrite step} && 
  \mbox{Single pushout rewrite step} \\
\end{array} $$

\noindent In the DPO approach \cite{EhrigPS73,CorradiniMREHL97}, a rule is
defined as a span $\rho = L\leftarrow K \rightarrow R$ and a match
is a morphism $m : L \rightarrow G$.  A graph $G$ rewrites into a
graph $H$ using rule $\rho$ and match $m$ if
the diagram above to the left can be constructed, where both squares are pushouts.
%
Conditions for the existence and uniqueness of graph $D$ need to be 
studied explicitly, since it is not  a universal construction.
With DPO rules  it is easy to specify  the addition, deletion, merging or 
cloning of items, but their applicability is limited.  For example, rule $\rho1$
of Figure \ref{fig_rulesAndGraphs} is not applicable to the grey node of $G$ (as it would leave
dangling edges), and a rule like $\rho2$ is usually forbidden as the \emph{pushout  complement}
$D$ would not be unique.  

In the SPO approach \cite{Lowe93,EhrigHKLRWC97}, a rule is a
\emph{partial} graph morphism $\psi: L \to R$ and a match is a
total morphism $m : L \rightarrow G$.  A graph $G$ rewrites into a
graph $H$ using rule $\psi$ and match $m$ if 
a square like the one above to the right can be constructed, which 
is a pushout in the category of graphs and partial morphisms. 
Deleting, adding and merging items can
easily be specified with SPO rules, and the approach is
appropriate for specifying deletion of nodes in unknown context, thanks to partial
morphisms. The deletion of a node causes  the deletion of all
edges connected to it, and thus applying 
rule $\rho1$ to  $G$ would result in graph $H1$ in Figure
\ref{fig_rulesAndGraphs}(b). However, since a rule is defined as a single
graph morphism, copying of items (as in rule $\rho2$) cannot be specified directly in SPO.



A more recent algebraic approach is the sesqui-pushout
approach (SqPO)~\cite{CorradiniHHK06}. Rules are
spans like in the DPO, but in the left square of a rewriting step, graph $D$ 
is built as a \emph{final pullback complement}. 
This characterises $D$ with a universal property, enabling to
apply rule $\rho1$, obtaining the same result as in the SPO approach
($H1$), as well as rule $\rho2$, obtaining $H2$ as  result. 
Also $\rho2$ has a   side effect: 
when a node is copied all the edges of the
original node are copied as well. 
Rules do not specify explicitly which context edges are deleted/copied, this is  determined
by the  categorical constructions that define rule application.  
In general, in all algebraic approaches, the items that are
preserved by a rule will retain the connections they have with items which are not in the image of the match. 
This holds also for items that are copied in the SqPO approach.

However, 
there are situations in
which the designer should be able to specify which of the
edges connecting the original node should be copied when a node is copied, 
depending for example on the direction of the edges (incoming or outgoing), or on 
their labels, if any. 
For example, if the graphs of Figure~\ref{fig_rulesAndGraphs} represent web 
pages (nodes) and hyperlinks among them (edges) it would be reasonable to expect that the 
result of copying the grey page of $G$ with rule $\rho2$ would be graph $H3$ rather than $H2$, so that new 
hyperlinks are created only in the new page, and not in the pages pointing to the original one.
%
As another example,   the {\tt fork} and {\tt clone} system commands in Linux both
generate a clone of a process, but with different semantics. Both commands
precisely differ in the way the environment of the cloned process is
dealt with: see \cite{MMitch01} for more details. 

These examples motivate the rewriting approach that we introduce in this paper.
In order to give the designer the possibility of controlling how the nodes that
are preserved or cloned by a rule  
are embedded in the context graph, we propose a new algebraic approach
to graph transformation where rules are triples of arrows with the
same source $(K \uprto{l} L, K \uprto{r} R, K \upmonoto{t} T_K)$.
Arrows $l$ and $r$ are the usual left- and right-hand sides,
while $t$ is a mono called the \emph{\embspec}:
it will play a role in
controlling which edges  from the context are copied. 
The resulting rewriting approach, called
\pbcpo{} (for Algebraic Graph Rewriting with controllEd Embedding) is presented in Sect.~\ref{sec:pbcpo}. 
As usual for the algebraic approaches, \agree{} rewriting will be introduced abstractly 
for a category satisfying suitable requirements, that will be introduced in Sect.~\ref{sec:preliminaries}. 
For the knowledgeable reader we anticipate that we will require the existence of \emph{partial map classifiers}~\cite{CL2}. After discussing an example of social networks in Sect.~\ref{sec:examples}, in Sect.~\ref{sec:sqpo} 
we show that \agree{} rewriting can simulate both SqPO rewriting (restricted to mono matches) and \emph{rewriting with polarised cloning}~\cite{DuvalEP12long}.
Finally some related and future
works are briefly discussed in Sect.~\ref{sec:discussion}. 
Two appendices collect the proofs of the main results, that were omitted in the published 
version~\cite{CorradiniDEPR-ICGT15}  of the present paper.

\section{Preliminaries}
\label{sec:preliminaries}

We start recalling some definitions and a few properties concerning pullbacks, partial maps and partial map classifiers: 
a survey on them can be found in  \cite{CL1,CL2}. 
Let $\catC$ be a category with all pullbacks.
We  recall the following properties:
\vspace{-0.3cm}
\begin{itemize}
\item monos are stable under pullbacks, i.e.~if  $B' \uplto{f'} A' \uprto{m'} A$ is the pullback of $B' \upmonoto{m} B \uplto{f} A$ and $m$ is mono, then $m'$ is mono as well.

\item the \emph{composition} property of pullbacks:
in a commutative diagram as  below on the left, 
if squares (a) and (b) are pullbacks,
so is the composed square;

$$\xymatrix@C=3pc@R=1pc{
\ar@{}[rd]|{PB~(a)} \bullet \ar[d] \ar[r] \ar@/^3ex/[rr]_{=} & 
  \ar@{}[rd]|{PB~(b)} \bullet \ar[d] \ar[r] & \bullet \ar[d] \\ 
\bullet \ar[r] \ar@/_3ex/[rr]^{=} & 
\bullet \ar[r] & \bullet \\ 
}
\qquad\qquad
\xymatrix@C=3pc@R=1pc{
\ar@{}[rd]|{PB~(c)} \bullet \ar[d] \ar@/^3ex/[rr]_{=} \ar@{-->}[r] & 
  \bullet \ar[d] \ar[r] \ar@{}[rd]|{PB~(d)} & \bullet \ar[d] \\  
\bullet \ar[r] & \bullet \ar[r] & \bullet \\ 
}
$$
\item and the \emph{decomposition} property:
in a commutative diagram as the one made of solid arrows above on the right, 
if square (d) and the outer square are pullbacks, then 
there is a unique arrow (the dotted one)  
such that the top triangle  commutes and square (c) is a pullback. 
\end{itemize}
%
 

\noindent A \emph{stable system of monos} of {\catC} is a family \M{} of monos including all isomorphisms, closed under composition, and (\emph{stability}) such that if $(f',m')$ is a pullback of $(m,f)$ and $m \in \M$, then $m' \in \M$. 
An \emph{\M-partial map} over  $\catC$, denoted $(m,f):Z\parto Y$, 
is a span made of a mono $m:X\monoto Z$ in \M{} and an arrow $f:X\to Y$ in \catC{}, 
up to the equivalence relation $(m',f')\sim(m,f)$ 
whenever there is an isomorphism $h$ with $m'\circ h = m$ 
and $f'\circ h = f$.

\noindent
\begin{minipage}{.6\textwidth}
Category \catC{} has an \emph{\M-partial map classifier} $(T,\eta)$ if $T$ is a functor $T: \catC{} \to \catC{}$ and $\eta$ is a natural transformation $\eta: Id_{\catC} \dotarrow T$, such that for each object $Y$ of \catC{}, the following holds:  for each 
\M-partial map $(m,f):Z\parto Y$ there is 
a unique arrow $\varphi(m,f):Z\to T(Y)$ such that   square~(\ref{pb:pmc}) is a pullback.

In this case it can be shown (see~\cite{CL2}) that $\eta_Y \in \M{}$
for each object $Y \in \catC$, that $T$ preserves pullbacks, and that the natural transformation $\eta$ is  \emph{cartesian},
which means that for each $f:X\to Y$ the naturality  
square~(\ref{pb:eta}) is a pullback.
For each mono $m:X\monoto Z$ in \M{}
we will use the notation $\olm=\varphi(m,\id_X)$,
thus $\olm$ is defined by the pullback square~(\ref{pb:olm}). 
\end{minipage}
\begin{minipage}{.4\textwidth}
\vspace{-4mm}
\begin{equation}
\label{pb:pmc} 
\xymatrix@C=5pc{
\ar@{}[rd]|{PB} X \ar@{ >->}[d]|{\bx{m}} \ar[r]|{\bx{f}} & Y \ar@{ >->}[d]|{\bx{\eta_Y}} \\ 
Z \ar[r]|{\bx{\varphi(m,f)}} & T(Y) 
} 
\end{equation}
\begin{equation}
\label{pb:eta} 
\xymatrix@C=4pc{
\ar@{}[rd]|{PB} X \ar@{ >->}[d]|{\bx{\eta_X}} \ar[r]|{\bx{f}} & 
  Y \ar@{ >->}[d]|{\bx{\eta_Y}} \\ 
T(X) \ar[r]|{\bx{T(f)}} & T(Y) \\ 
} 
\end{equation}
\begin{equation}
\label{pb:olm} 
\xymatrix@C=5pc{
\ar@{}[rd]|{PB} X \ar@{ >->}[d]|{\bx{m}} \ar[r]|{\bx{\id_X}} & X \ar@{ >->}[d]|{\bx{\eta_X}} \\ 
Z \ar[r]|{\bx{\olm}} & T(X) \\ 
} 
\end{equation}
\end{minipage}

Before discussing some examples of categories that have \M{}-partial map classifiers, let us recall the definition of some categories of graphs.

\begin{definition}[graphs, typed graphs]
\label{de:graphs}
 The category of \emph{graphs} $\Gr$ is defined as follows.  
A \emph{graph} $X$ is made of a set of \emph{nodes} $N_X$, a set
of \emph{edges} $E_X$ and two functions $s_X,t_X:E_X\to N_X$, called \emph{source} and
\emph{target}, respectively. As usual, we write $n\uprto{e}p$ when $e \in E_X$, $n = s_X(e)$ and $p = t_X(e)$.
A \emph{morphism} of graphs $f:X\to Y$ is made of two functions 
$f:N_X\to N_Y$ and $f:E_X\to E_Y$, such that
$f(n)\uprto{f(e)}f(p)$ in $Y$ for each edge $n\uprto{e}p$ in $X$. 

 Given a fixed graph $\mathit{Type}$, called \emph{type graph}, the category of \emph{graphs typed over} $\mathit{Type}$ is the slice category $\Gr \slice \mathit{Type}$.
 \end{definition}

\begin{definition}[polarized graphs~\cite{DuvalEP12}]
\label{def:pol-cat}
A \emph{polarized graph} $\grpol{X}=(X,N_X^+,N_X^-)$ is a graph $X$ 
with a pair $(N^+,N^-)$ of subsets of the set of nodes $N_X$
such that for each edge $n\uprto{e}p$ 
one has $n\in N_X^+$ and $p\in N_X^-$.  
A \emph{morphism} of polarized graphs $f: \grpol{X} \to \grpol{Y}$, 
where $\grpol{X}=(X,N_X^+,N_X^-)$ and $\grpol{Y}=(Y,N_Y^+,N_Y^-)$,
is a morphism of graphs $f:X\to Y$ such that
$f(N_X^+)\subseteq N_Y^+$ and $f(N_X^-)\subseteq N_Y^-$.
This defines the \emph{category} $\Grpol$ of polarized graphs.

A morphism of polarized graphs $f: \grpol{X} \to \grpol{Y}$ is \emph{strict},
or \emph{strictly preserves the polarization}, 
if $f(N_X^+)= f(N_X)\cap N_Y^+$ and 
$f(N_X^-)= f(N_X)\cap N_Y^-$.
\end{definition}

\vspace{-0.5cm}

\subsection{Examples of Partial Map Classifiers}
\label{sect:examples-classifiers}
Informally, if $(m,f):Z\parto Y$ is a partial map, a total arrow $\varphi(m,f):Z \to T(Y)$
representing it should agree with $(m,f)$ on the ``items'' of $Z$ on which it is defined,
and should map any item of $Z$ on which $(m,f)$ is not defined in 
a unique possible way to some item of $T(Y)$ which does not belong to (the image via $\eta_Y$ of) $Y$.
For example, in \Set{} the partial map classifier $(T,\eta)$ 
is defined as
 $T(X)=X+\{\ourbot\}$ and $T(f)=f+\id_{\{\ourbot\}}$ for functor $T$, while the 
 natural transformation $\eta$ 
 is made of the inclusions $\eta_X:X\to X+\{\ourbot\}$.
For each partial function $(m,f):Z\parto Y$, 
function $\varphi(m,f):Z\to Y+\{\ourbot\}$
extends $f$ by mapping $x$ to $f(x')$ when $x=m(x')$ 
and $x$ to $\ourbot$ when $x$ is not in the image of $m$.



 \vspace{-0.3cm}
  \begin{figure}[htbp]
\begin{center}
\includegraphics[width=\textwidth]{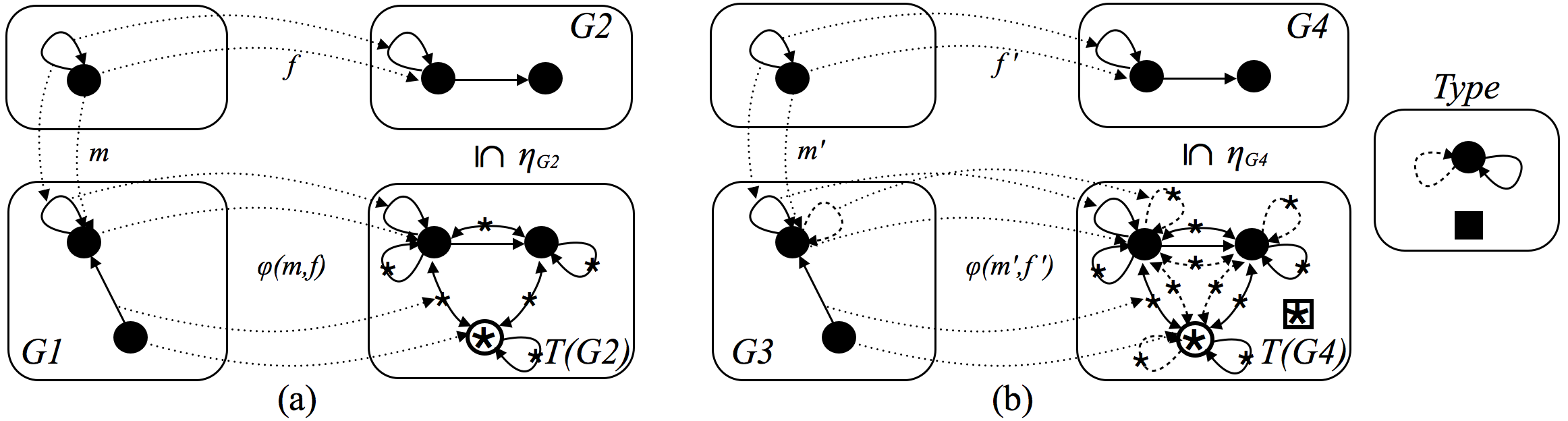}\vspace{-0.5cm}
\end{center}
\caption{Partial Map Classifiers (a) in $\Gr$ (b) in $\Gr\downarrow Type$}
\label{fig_PartialMapClassExamples}
\end{figure}
\vspace{-0.3cm}

In  $\Gr$  the partial map classifier $(T, \eta)$
is such that 
$\eta_G: G \to T(G)$
embeds $G$ into the graph $T(G)$ 
made of the disjoint union of $G$ with a node $\ourbot$ 
and with an edge $\ourbot_{n,p}: n\to p$ for each pair 
of vertices $(n,p)$ in $(N_G+\{\ourbot\}) \times (N_G+\{\ourbot\})$.  
The total morphism $\varphi(m,f)$ 
is defined on the set of nodes exactly as in {\Set}, and on each edge similarly, but consistently with the way its source and target nodes are mapped.
Figure \ref{fig_PartialMapClassExamples}(a) shows an example of a partial map $(m,f): \mathit{G1}\to  \mathit{G2}$ and the corresponding extension to the total morphism $\varphi (m,f): \mathit{G1}\to \mbox{\emph{T(G2)}}$. In the graphical notation we use edges with double tips to denote two edges, one in each direction; arrows and node marked with $*$ are added to $ \mathit{G2}$ by the $T$ construction.


$\Set$ and $\Gr$ are instances of the general result that all elementary toposes have \M{}-partial map classifier, for \M{} the family of all monos. These include, among others, all \emph{presheaf categories} (i.e., functor
categories like $\Set^{\catC^{\mathrm{op}}}$, where $\catC$ is a small category), and the slice
categories like $\catC \slice X$ where $\catC$ is a topos and $X$ an object of $\catC$.
In fact $\Gr$ is the presheaf category $\Set^{\catC^{\mathrm{op}}}$ where $\catC^{\mathrm{op}}$ has two objects $E$, $N$ and two non-identity arrows $s, t: E \rto N$.

As a consequence also the category of typed graphs  $\Gr\downarrow \mathit{Type}$ has partial maps classifiers for all monos. Figure \ref{fig_PartialMapClassExamples}(b) shows an example: the partial map classifier of a graph $G4$ typed over $\mathit{Type}$ is obtained by adding to $G4$ all the nodes of $\mathit{Type}$ and, for each pair of nodes of the resulting graph, one instance of each edge that is compatible with the type graph.


The category of \emph{polarized graphs} of Def.~\ref{def:pol-cat}   (that will be used later in Sect.~\ref{subsec:pbcpovspolclo}), is an example of category which has \M{}-partial map classifiers for a family \M{} which is a proper subset of all monos.
It is easy to check that strict monos form a stable system of monos (denoted \Str{}) for category $\Grpol$, and that 
$\Grpol$ has an \Str{}-partial map classifier $(\TT, \eta)$.  Morphism $\eta_{\grpol{K}}$ embeds a polarised graph $\grpol{K}$ into $\TT(\grpol{K})$, which is the disjoint union of
$\grpol{K}$ with a node $\ourbot$ (having polarity $\pm$) and with an edge $\ourbot_{n,p} :
n\to p$ for each pair of nodes 
$(n,p)\in(N_K^+ +\{\ourbot\})\times (N_K^- +\{\ourbot\})$. 
The total morphism $\varphi(m,f)$ is defined exactly as in the category of graphs.

\section{Algebraic Graph Rewriting with Controlled Embedding}
\label{sec:pbcpo}

In this section we introduce the {\pbcpo} approach to rewriting,
defining rules, matches and rewrite steps.  The main difference with
respect to the DPO and SqPO approaches is that a rule has an
additional component $t: K \monoto T_K$, called the {\emph {\embspec}}, that
enriches the interface and can be used to control the embedding of
preserved items.
We assume that  $\catC$ is a category with all pullbacks, 
with a stable system of monos \M{}, with an \M{}-partial map classifier $(T,\eta)$, and with pushouts along monos in \M{}. 


\begin{definition}[{\pbcpo} rules and matches]
\label{def:pbcpo}

\noindent
\begin{minipage}{.7\textwidth}
  {\ --\ }A \emph{rule} is a triple of arrows with the same source 
  $\rho = (K \uprto{l} L, K \uprto{r} R, K \upmonoto{t} T_K)$, with $t$ in \M{}. Arrows  $l$ and $r$ are the \emph{left-} and \emph{right-hand side}, respectively, and $t$ is called the  \emph{\embspec{}}. 
\end{minipage}
\begin{minipage}{.3\textwidth}
\mbox{}
\vspace{-1cm}
$$ \xymatrix@C=2pc@R=1.5pc{
L & K \ar[l]_{l} \ar[r]^{r} \ar@{ >->}[d]^{t} & R \\  
& T_K & }$$ 
\end{minipage}
\noindent
{\ --\ } A \emph{match} of a rule $\rho$ with left-hand-side $K \uprto{l} L$ 
  is a mono $L \upmonoto{m} G $ in \M{}. 
\end{definition}
 
\vspace{-0.3cm}
\begin{equation}
 \label{eq:agree-rew}
 \xymatrix@C=6pc{
\ar@{}[dr]|{PB~(remark)} L \ar@{ >->}[d]^{m} 
  \ar@<-.5ex>@{ >->}@/_3ex/[dd]_(.6){\eta_L}^(.6){=} & 
  \ar@{}[dr]|{PO~(b)} K \ar[l]_{l} \ar[r]^{r} \ar@{ >->}[d]_{n} 
    \ar@<.5ex>@{ >->}@/^3ex/[dd]^(.6){t}_(.6){=}|{\hole} & 
    R \ar[d]^{p} \\
\ar@{}[dr]|{PB~(a)} G \ar[d]^{\olm} & 
  D \ar[l]_{g} \ar[r]^{h} \ar[d]_{n'} & 
    H \\ 
T(L) & 
  T_K \ar[l]^{l' = \varphi(t,l)} & 
  }
\end{equation}
\vspace{-0.3cm}

 \begin{definition}[\agree{} rewriting]
 \label{def:pbcpo-rewriting}
 Given a rule $\rho =  (K \uprto{l} L, K \uprto{r} R, K \upmonoto{t} T_K)$ and a match $L \upmonoto{m} G$, an {\agree} \emph{rewrite step} $G \Rightarrow_{\rho,m} H$ is constructed in two phases as follows (see diagram~(\ref{eq:agree-rew})):
\\ \noindent
{\ (a) \ } Let  $l'=\varphi(t,l): T_K \to T(L)$ and $\olm=\varphi(m,\id_L):G \to T(L)$,
then $G \uplto{g} D \uprto{n'} T_K$ is the pullback of $G \uprto{\olm} T(L) \uplto{l'} T_K$. 
\\ \noindent
{\ (remark) \ } In diagram~(\ref{eq:agree-rew}) $(g,n')$ is a pullback of $(\olm,l')$ 
and $(l,t)$ is a pullback of $(\eta_L,l')$ because $l'=\varphi(t,l)$,
thus by the decomposition property there is a unique $n:K\to D$
such that $n'\circ n=t$, $g\circ n = m\circ l$ and 
$(l,n)$ is a pullback of $(m,g)$. 
Therefore $n$ is a mono in $\M{}$ by stability. 
\\ \noindent
{\ (b) \ } Let $n$ be as in the previous remark.
Then $R \uprto{p} H \uplto{h} D$ is the pushout of $D \uplto{n} K \uprto{r} R$.
 \end{definition}
 
 \begin{example}
Using the AGREE approach, the web page copy operation can be modelled using
 the rule $(K1 \to L1, K1 \to R1, K1 \monoto TK1)$
shown in Figure \ref{fig_webpagecopy}. This rule
is typed over the type graph $\mathit{Type}$. Nodes denote web pages, solid edges denote links and dashed edges describe the subpage relation. 
The different node colours (gray and black) are used just to define the match, whereas
 the {\bf{\textsf{c}}} inside some nodes is used to indicate that this is a
 copy. When this rule is applied to graph ${G1}$, only out-links are
 copied because the pages that link the copied one remain the
 same, that is, they only have a link to the original page, not to
 its copy. The subpage structure is not copied. 
 Note that all black nodes of ${G1}$ and ${D1}$ are mapped to $*$-nodes of $T(L1)$ and ${\mathit{TK}1}$, respectively.
 \vspace{-0.5cm}
  \begin{figure}[htbp]
\begin{center}
\includegraphics[width=0.8\textwidth]{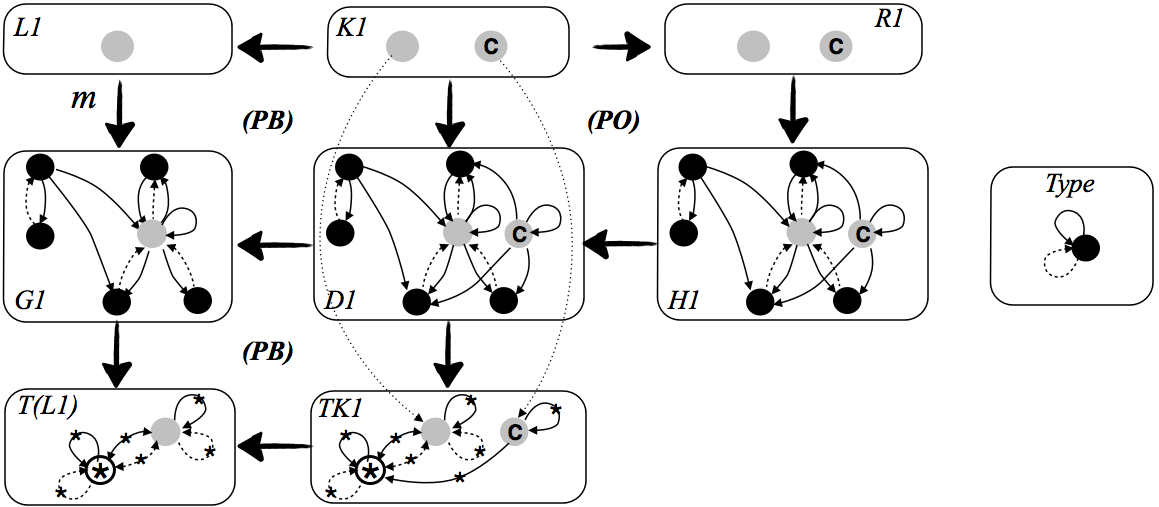}\vspace{-0.5cm}
\end{center}
\caption{Rule for copying a web page and example of application}
\label{fig_webpagecopy}
\end{figure}
\vspace{-0.7cm}
\end{example}


In the general case just presented, the \embspec{}  $t$ could have a non-local effect on the rewritten object.
In the following example, based on category {\bf Set}, the rule simply preserves a single element and $t: K \to T_K$ is the identity. 
If applied to set $G$, its effect is to delete all the elements not matched by $m$, as shown. 
We say that this rewrite step is \emph{non-local}, because it modifies the complement of the image of $L$ in $G$.

\begin{center}
\includegraphics[width=0.6\textwidth]{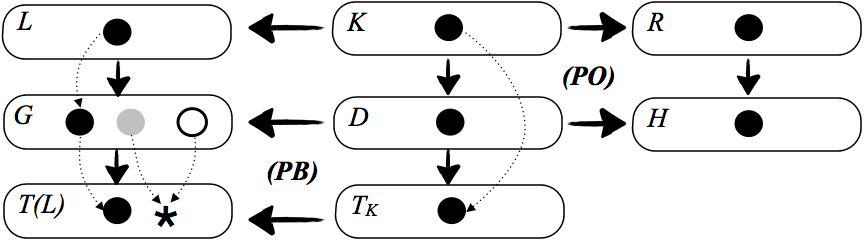}
\end{center}


 
In the rest of this section we present a condition on rules that 
ensures the locality of the rewrite steps. In order to formulate this condition in the general setting 
of a  category with \M{}-partial map classifiers, we need to  consider a generalisation of the notion 
of complement of a subset in a set, that we call  \emph{strict complement}. 
For instance, in  category $\Gr$, the strict complement of a subgraph $L$
in a graph $G$ is the largest subgraph $G\setminus L$ of $G$ disjoint from $L$;
thus, the union of $L$ and $G\setminus L$ is in general smaller than $G$.
%
Intuitively, we will say that an \agree{} rewrite step as in diagram~(\ref{eq:agree-rew}) is \emph{local}
if the strict complement of $L$ in $G$ is preserved, i.e., if $g$ restricts to an isomorphism 
between $D\setminus K$ and $G\setminus L$.  

For the definitions and results that follow, we assume that category $\catC$, besides satisfying the conditions listed at the beginning of this section, has a final object~$\fin$ and a \emph{strict} initial object $\ini$ (i.e., each arrow with target $\ini$ must have~$\ini$ as source); furthermore, the unique arrow from $\ini$ to $\fin$, that we denote $\inifin: \ini \to \fin$, belongs to \M{}.
For each object $X$ of $\catC$ we will denote by $\final_X:X\to \fin$ the unique arrow to the final object,  and by $\init_X:\ini\to X$  the unique arrow from the 
initial object. 

\hide{Since $\fin$ is a terminal object,  
each arrow with source $\fin$ is a mono.} 
\hide{Let us start by observing that since $\ini$ is a \emph{strict} initial object, each arrow with source $\ini$ is a mono,
including the only arrow from $\ini$ to $\fin$ that we denote $\inifin:\ini\to\fin$. }
For each mono $m:L\monoto G$ in \M{} the \emph{characteristic arrow} of $m$
is defined as $\chi_m = \varphi(m,\final_L):G\to T(\fin)$,
(see pullback (a) in diagram (\ref{pb:sub-class-new})). 
Object $T(\fin)$ is called the \emph{\M{}-subobject classifier}. 

\begin{equation}
\label{pb:sub-class-new} 
\xymatrix@C=6pc{
\ar@{}[rd]|{PB~(c)} K \ar@{ >->}[d]|{\bx{n}} \ar[r]|{\bx{l}} & 
\ar@{}[rd]|{PB~(a)} L \ar@{ >->}[d]|{\bx{m}} \ar[r]|{\bx{\final_L}}  & 
\final \ar@{ >->}[d]^{\eta_{\final}}_{\ttrue} \\ 
\ar@{}[rd]|{PB~(d)}  D \ar[r]|{\bx{g}} &
\ar@{}[rd]|{PB~(b)}  G \ar[r]|{\bx{\chi_m = \varphi(m,\final_L)}} & T(\final) \\ 
D\setminus K \ar@{->}[u]|{\bx{D\setminus n}} \ar[r]|(.5){\bx{g\setminus l}} &
G\setminus L \ar@{->}[u]|{\bx{G\setminus m}} \ar[r]|(.6){\bx{\fin_{G\setminus L}}} &
\fin \ar@{->}[u]^{{\ffalse}}_{T(\inifin)\circ  \olinifin}
} 
\end{equation}

By exploiting the assumption that $! \in \M{}$ and that $\ini$ is
strict initial, 
it can be shown that 
$T(\ini)$ is isomorphic to $\fin$,  with $\olinifin = \final_{T(\ini)}^{-1}$, and this yields an arrow 
$T(\inifin) \circ \olinifin: \fin \to T(\fin)$.
In category \Set{}
(with \M{} the family of all injective functions) arrows $\eta_1$ and $T(\inifin) \circ \olinifin: 1 \to T(1)$ are the coproduct injections of the subobject classifier (which is a two element set), and are also known as $\ttrue$
and $\ffalse$, respectively.   In \Set{}  the complement of an injective function $m : L \monoto G$ can be defined as the pullback of $\chi_m: G \to T(1)$ along $\ffalse$.  We generalise this to the present setting as follows.

\begin{definition}[strict complements]
\label{def:complement} 
Let  $\catC$ be a category that satisfies the conditions listed at the beginning of Section~\ref{sec:pbcpo}, has final object $1$, strict initial object $0$, and such that 
$\inifin \in \M{}$. Let  $m:L\monoto G$ be a mono in \M{}, and $\chi_m: G \to T(\final)$ be its characteristic arrow defined by pullback (a) of diagram~(\ref{pb:sub-class-new}). Then the \emph{strict complement of $L$ in $G$ (with respect to $m$)} is the arrow $G\setminus m: G\setminus L\monoto G$ obtained 
as the pullback of $\chi_m$ and $\ffalse{}  =T(\inifin)\circ \olinifin: 1 \to T(\final)$, as in square (b) of diagram~(\ref{pb:sub-class-new}).

Furthermore, for each  pair of monos $n:K\monoto D$ and $m:L\monoto G$ in \M{} and for each pair of arrows $l: K \to L$ and $g: D \to G$ such that square~(c) of diagram~(\ref{pb:sub-class-new}) is a pullback,  
arrow $g\setminus l: D\setminus K \rto G\setminus L$ as in square~(d) is called the \emph{strict complement} of $l$ in $g$ (with respect to $n$ and $m$). 
\end{definition}

It is easy to check that arrow $g\setminus l$ exists and is uniquely determined by the fact that square~(b) is a pullback; furthermore square~(d) is a pullback as well, by decomposition. We will now exploit the notion of strict complement to formalize locality of \agree{} rewriting.

\begin{definition}[local rules and local rewriting in \agree{}] 
\label{def:embSpec} 
An {\agree} rule $\rho = (l,r,t)$ is \emph{local} if $\ol{t}:T_K\to T(K)$ 
is such that $\ol{t}\setminus \id_K:T_K\setminus K \to T(K)\setminus K$ is an iso. 
An \agree{} rewrite step as in diagram~(\ref{eq:agree-rew}) is \emph{local} if arrow 
$g\setminus l: D\setminus K \rto G\setminus L$ is an iso.
\end{definition}

The definition of local rewrite steps is as expected, but that of local rules deserves some comments. 
Essentially, in the first phase of \agree{} rewriting, when building the pullback~(a) of diagram~(\ref{eq:agree-rew}),  the shape of $T_K \setminus K$  determines the effect of the rule on the strict complement of $L$ in $G$, which is mapped by $\olm$ to $T(L) \setminus L$.  It can be proved that $T(L) \setminus L$ is isomorphic to 
$T(K) \setminus K$, therefore if the rule is local we have that  $T_K \setminus K$ is isomorphic to $T(L) \setminus L$, and this guarantees that the strict complement of $L$ in $G$ is preserved in the rewrite step.
These considerations provide an outline of the proof of the main result of this section, which is reported in Appendix~\ref{sec:AppProofs}.

\begin{proposition}[locality of {\agree} rewrite steps]
\label{prop:locality}
Let $\rho = (l,r,t)$ be a local rule.
Then, with the notations as in diagram~(\ref{eq:agree-rew}), for each match $L \upmonoto{m} G$ the resulting rewrite step $G \To_{\rho,m} H$ is local.
\end{proposition}
%

\section{Example: Social Network Anonymization}
\label{sec:examples}

Huge network data sets, like social networks (describing personal
relationships and cultural preferences) or communication networks (the
graph of phone calls or email correspondents) become more and more
common. These data sets are analyzed in many ways varying from the
study of disease transmission to targeted advertising. Selling network
data set to third-parties is a significant part of the business model
of major internet companies. Usually, in order to preserve the
confidentiality of the sold data set, only ``anonymized'' data is
released. The structure of the network is preserved, but personal
identification informations are erased and replaced by random
numbers. This anonymized network may then be subject to further
processing to make sure that it is not possible to identify the nodes
of the network (see \cite{Hay10} for a discussion about
re-identification issues). We are going to show how AGREE rewriting
can be used for such anonymization procedure. Of course, due to space
limitations we cannot deal with a complete example and will focus on the
first task of the anonymization process: the creation of a clone of
the social network in which only non-sensitive links are
copied. We model the following idealized scenario: the administrator of
a social network sells anonymized data sets to third-parties so that
they can be analyzed without compromising confidentiality. Our graphs
are made of four kinds of nodes: customer  (grey nodes),
administrator of the social network (white node), user of the
social network  (black nodes) and square nodes that model the
fact that data will suffer post-processing. Links of the social
network can be either public (black solid) or private (dashed -- this
latter denotes sensitive information that should not be disclosed),
moreover we use another type of edges (grey), denoting the fact that a
node ``knows", or has access to another node. The corresponding type
graph $Type$ is shown in Figure \ref{fig_graphs}. 
 \begin{figure}[htbp]
\begin{center}
\includegraphics[width=0.2\textwidth]{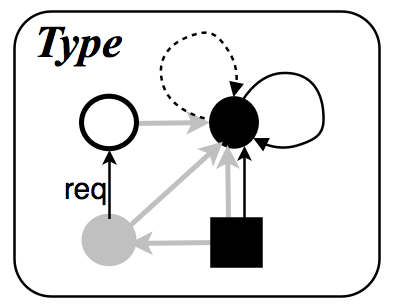}
\includegraphics[width=0.55\textwidth]{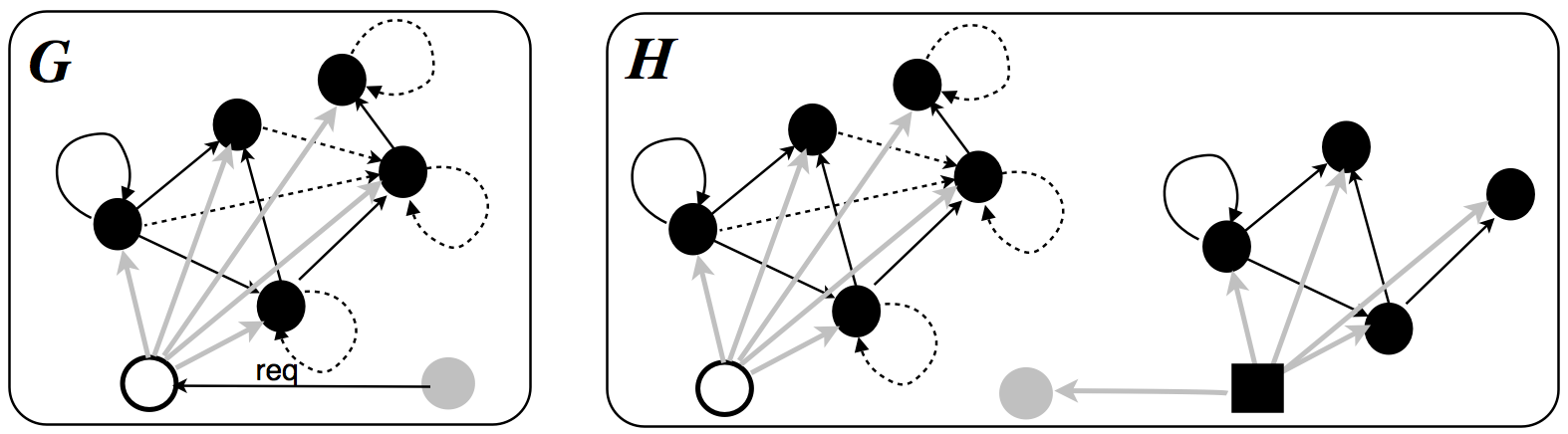}
\vspace{-0.5cm}
\end{center}
\caption{Type Graph $Type$, Graphs $G$ and $H$}
\label{fig_graphs}
\end{figure}

The rule depicted in Figure \ref{fig_anonymizerule} shows an example
that anonymizes a portion of a social network with $4$ nodes
(typically portions of a fixed size are sold).  
Graph $TK$  
consists of a clique of all copies of matched black nodes (denoted by {\bf{\textsf{c}}}) with public links, and
a graph representing the $T$ construction applied to the rest of $K$.
To enhance
readability, we just indicated that the graph inside the dotted square
should be completed according to $T$: a copy
of the nodes of the type graph should be added, together with all
possible edges that are compatible with the type graph.
This allows
the cloning of the subgraph defined by the match limited to public
edges. In the right hand side $R$ a new square node is added marking
the cloned nodes for post-processing.  The application of this rule to
graph $G$ in Figure \ref{fig_graphs} with a match not including the top black node produces graph $H$. 
 \begin{figure}[htbp]
\begin{center}
\includegraphics[width=0.8\textwidth]{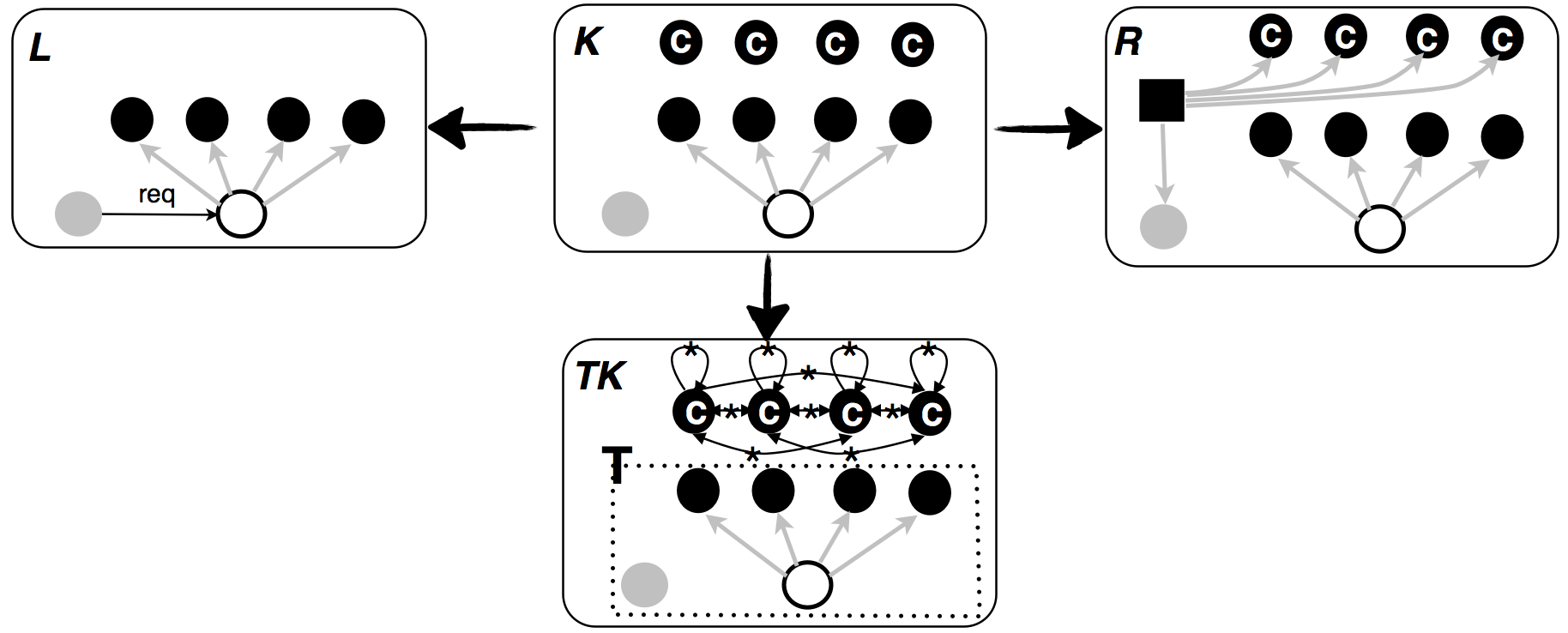}
\vspace{-0.5cm}
\end{center}
\caption{4-Anonymize rule}
\label{fig_anonymizerule}
\end{figure}
\vspace{-1cm}


\section{\pbcpo  Subsumes \sqpo  and Polarized Node Cloning}
\label{sec:sqpo}

As recalled in the Introduction, in the \sqpo{} approach~\cite{CorradiniHHK06}  a rule is a span 
$L \uplto{l} K \uprto{r} R$ and a rewriting step for a match $L \uprto{m} G$ is made of 
a first phase where the \emph{final pullback complement} $D$ is constructed, and next a pushout with the right-hand side is performed.  

\begin{definition}[final pullback complement]
\label{def:fpbc}
In  diagram~(\ref{eq:fpbc}), $K \uprto{n} D \uprto{a} G$ is a \emph{final pullback complement} of 
$K \uprto{l} L \uprto{m} G$ if 

\noindent
\begin{minipage}{.6\textwidth}
\begin{enumerate}
\item the resulting square is a pullback, and 
\item for each pullback 
$G \uplto{m} L \uplto{d} K' \uprto{e} D' \uprto{f} G$ and  arrow $K' \uprto{h} K$ such that 
$l \circ h = d$,  there is a unique arrow  $D' \uprto{g} D$ such that $a \circ g = f$ and $g \circ  e = n \circ h$.
\end{enumerate}
\end{minipage}
\begin{minipage}{.4\textwidth}
\begin{equation}
\label{eq:fpbc}
 \xymatrix@C=1pc@R=0.5pc{
L \ar[dd]|{\bx{m}}   
& & K \ar[ll]|{\bx{l}}     \ar[dd]|{\bx{n}} 
& & K'   \ar[ll]|{\bx{h}}  \ar[dd]|{\bx{e}}  \ar@/_3ex/[llll]|{\bx{d}}   \\
&  \\
G  
& & D   \ar[ll]|{\bx{a}}
& & D'   \ar@{-->}[ll]_{\bx{g}}  \ar@/^3ex/[llll]|{\bx{f}}
}
\end{equation}
\end{minipage}
\end{definition}

\noindent
The next result shows that in a  category with a stable system of monos \M{} and with $\M{}$-partial map classifiers, the final pullback complement of $m \circ l$, with $m \in \M{}$, can be obtained by taking the pullback of $T(l)$ along $\olm$. This 
means that if the embedding morphism of an  \agree{} rule is the partial map classifier of  $K$, i.e., $K \upmonoto{\eta_K} T(K)$, then the first phase of the \agree{} rewriting algorithm of Definition~\ref{def:pbcpo-rewriting} actually builds the final pullback complement of the left-hand side of the rule and of the match. This will allow us to relate the \agree{} approach with others based on the construction of final pullback complements. 

\begin{theorem}[building final pullback complements] 
\label{theorem:sqpo}
Let $\catC$ be a category with pullbacks, 
with a stable system of monos \M{} and with an \M{}-partial map classifier $(T,\eta)$.
Let   $K \uprto{l} L$ be an arrow in \catC{} and $L \upmonoto{m} G$ be a mono in \M{}. Consider the naturality square built over  $K \uprto{l} L$ on the left of Figure~\ref{fig:FPBCasPB}, which is a pullback because $\eta$ is cartesian, and 
let $G \uplto{a} D \uprto{n'} T(K)$ be the pullback of $G \uprto{\olm}
T(L) \uplto{T(l)} T(K)$. Then $K \uprto{n} D \uprto{a} G$ is a final
pullback complement of   $K \uprto{l} L \uprto{m} G$, where $n$ is the only arrow 
making the right triangle commute and the top square a pullback.
\hide{as observed in Fact~\ref{fact:rule-match}.} 
\end{theorem}
\vspace{-.5cm}
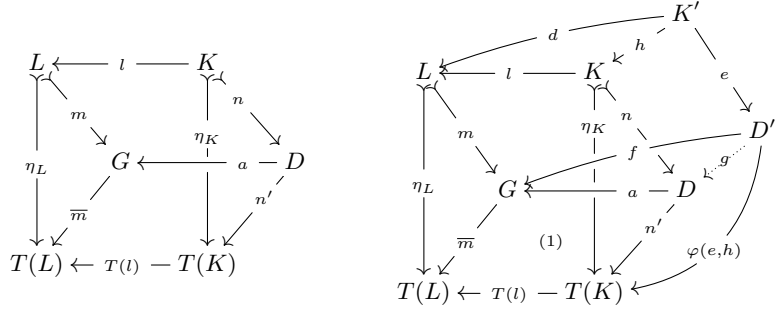
\begin{figure}
\hfill 
$ \xymatrix@C=1pc@R=0.8pc{
& & &\\
L \ar@{ >->}[ddr]|{\bx{m}}     \ar@{ >->}[dddd]|{\bx{\eta_L}} 
& & K \ar[ll]|{\bx{l}}    \ar@{ >->}[dddd]|(.35){\bx{\eta_K}}|(.5){\hole}    \ar@{ >->}[ddr]|(.35){\bx{n}} \\
& & &  \\
& G  \ar[ddl]|{\bx{\overline{m}}}  
& & D  \ar[ddl]|(.35){\phantom{\big(}n'\phantom{\big(}}  \ar[ll]|(.3){\bx{a}}\\
& \\
T(L)  
& & T(K) \ar[ll]|{\bx{T(l)}} 
}$
\hfill
$ \xymatrix@C=1pc@R=0.8pc{
& & & K'   \ar@/_1ex/[dlll]|{\bx{d}}    \ar[ddr]|{\bx{e}}  \ar[dl]|{\bx{h}}   \\
L \ar@{ >->}[ddr]|{\bx{m}}     \ar@{ >->}[dddd]|{\bx{\eta_L}} 
& & K \ar[ll]|{\bx{l}}    \ar@{ >->}[dddd]|(.23){\bx{\eta_K}}|(.4){\hole}|(.53){\hole}    \ar@{ >->}[ddr]|(.35){\bx{n}}|(.63){\hole} \\
& & & & D'   \ar@{..>}[dl]|{{g}}   \ar@/_1ex/[dlll]|{\bx{f}}   \ar@/^5ex/[dddll]|{\phantom{\big(}\varphi(e,h)\phantom{\big(}}   \\
& G  \ar[ddl]|{\bx{\overline{m}}}  \ar@{}[ddr]|{(1)}   
& & D  \ar[ddl]|(.35){\phantom{\big(}n'\phantom{\big(}}  \ar[ll]|(.3){\bx{a}}\\
& \\
T(L)  
& & T(K) \ar[ll]|{\bx{T(l)}} 
}$\hfill\mbox{}
\caption{Constructing the final pullback complement of $m \circ l$ with a pullback.}
\label{fig:FPBCasPB}
\end{figure}
\vspace{-.5cm}
\begin{proof}
By the decomposition property we have that $K \upmonoto{n} D  \uprto{a} G$ is a pullback complement of  $K \uprto{l} L  \upmonoto{m} G$, and $n \in \M{}$ by stability. We have to show that the pullback complement is final, i.e.~that 
%
given a pullback  $G \uplto{m} L \uplto{d} K' \uprto{e} D' \uprto{f} G$  and an arrow   $K' \uprto{h} K$ such that  $l \circ h = d$, as shown on the right of Figure~\ref{fig:FPBCasPB}, there is a unique arrow $D' \uprto{g} D$ such that $n \circ h = g \circ e$ and $a \circ g = f$. We present here the \emph{existence} part, while the proof of \emph{uniqueness} is in Appendix \ref{sec:AppProofs}.

 Note that $K' \upmonoto{e} D'$ is in \M{} by stability. 
By the properties of the \M{}-partial map classifier $T$, there is a unique arrow $D' \uprto{\varphi(e,h)} T(K)$ such that $\eta_K \circ h = \varphi(e,h) \circ e$ and the square is a pullback. We will show below that $ \olm \circ f = T(l) \circ \varphi(e,h)$, hence by the universal property of the pullback $(1)$ there is a  unique arrow $D' \uprto{g} D$ such that $n' \circ g = \varphi(e,h)$ and $a \circ g = f$. It remains to show that $n \circ h = g \circ e$: by exploiting again pullback $(1)$, it is sufficient to show that (i) $a \circ n \circ h = a \circ g \circ e$ and (ii) $n' \circ n \circ h = n' \circ g \circ e$. In fact we have, by simple diagram chasing:
\begin{description}
\item (i) $a \circ n \circ h =  m \circ l \circ h = m \circ d = f \circ e = a \circ g \circ e$
\item (ii) $n' \circ n \circ h = \eta_K \circ h =  \varphi(e,h) \circ e =    n' \circ g \circ e$
\end{description}

We still have to show that $ \olm \circ f = T(l) \circ \varphi(e,h)$. This follows by comparing the following two diagrams, where all squares are pullbacks, either by the statements of Section~\ref{sec:preliminaries} or (the last to the right) by assumption.
Clearly, also the composite squares are pullbacks, but then the bottom arrows must both be equal to $\varphi(e,d)$, as in Equation~(\ref{pb:pmc}). Therefore we conclude that  $ \olm \circ f = 
\varphi(e,d) = T(l) \circ \varphi(e,h)$. 
$$
 \xymatrix@C=4pc@R=3pc{
L \ar@{ >->}[d]|{\bx{\eta_L}} \ar@{}[rd]|{PB~(\ref{pb:eta})} & 
K \ar[l]|{\bx{l}} \ar@{ >->}[d]|{\bx{\eta_K}} \ar@{}[rd]|{PB~(\ref{pb:pmc})}& 
K' \ar[l]|{\bx{h}} \ar@{ >->}[d]|{\bx{e}} \ar@/_4ex/[ll]|{\bx{d}} 
\\
T(L) & 
T(K) \ar[l]|{\bx{T(l)}}& 
D' \ar[l]|(.5){\bx{\varphi(e,h)}}
}
\quad\quad
\xymatrix@C=4pc@R=3pc{
L \ar@{ >->}[d]|{\bx{\eta_L}} \ar@{}[rd]|{PB~(\ref{pb:olm})} & 
L \ar[l]|{\bx{id_L}} \ar@{ >->}[d]|{\bx{m}}& 
K' \ar[l]|{\bx{l \circ h}} \ar@{ >->}[d]|{\bx{e}} \ar@/_4ex/[ll]|{\bx{d}} 
\\
T(L) & 
G \ar[l]|{\bx{\olm}}& 
D' \ar[l]|(.5){\bx{f}}}
$$
\end{proof} 

The statement of Theorem~\ref{theorem:sqpo} can be formulated equivalently in a more abstract way, as the fact that composing functor $T$ with a pullback along $\olm$ one gets a functor that is right adjoint to the functor taking pullbacks along $m$. This alternative presentation and its proof are presented in Appendix~\ref{app:abstract}.

\subsection{{\agree} subsumes \sqpo  rewriting with injective matches}
\label{subsec:pbcpovssqpo}

Using Theorem~\ref{theorem:sqpo} it is easy to show that  the \agree{} approach is a conservative extension of the \sqpo{} approach, because the two coincide if the embedding of the \agree{} rule is the arrow injecting $K$ into its partial map classifier.
 
\begin{theorem}[\agree{} subsumes \sqpo{} with monic matches]
\label{theorem:agreeVSsqpo}
Let $\catC$ be a category with all pullbacks, 
with \M{}-partial map classifiers $\eta: Id_{\catC} \dotarrow T$ for a stable system of monos \M{}, 
and with pushouts along arrows in \M{}.  
Let $\rho = L \uplto{l} K \uprto{r} R$ be a rule and $m:L\monoto G$
be a match in \M{}. Then 
$$G \To_{\rho,m}^{\sqpo{}} H \qquad \mbox{ if and only if } \qquad G \To_{(l,r,\eta_K),m}^{\agree{}} H$$
\noindent In words, the application of rule $\rho$ to match $m$ using the \sqpo{} approach has exactly the same effect of applying to $m$ the same rule enriched with the embedding $K \upmonoto{\eta_K} T(K)$ using the \agree{} approach.
\end{theorem}
\begin{proof}
Since the embedding of the rule is arrow $\eta_K: K \monoto T(K)$, phase (a) of \agree{} rewriting (Definition~\ref{def:pbcpo-rewriting})  is exactly the construction that is shown, in Theorem~\ref{theorem:sqpo}, to build  $K \uprto{n} D \uprto{a} G$ as a final pullback complement of  $K \uprto{l} L \uprto{m} G$, therefore it coincides with the construction of the left square of the {\sqpo{}} approach. The second phase, i.e.~the construction of the pushout of  $K \uprto{n} D$ and  $K \uprto{r} R$ is identical for both approaches by definition.
\end{proof}


\subsection{{\agree} subsumes polarized node cloning on graphs} 
\label{subsec:pbcpovspolclo}

We now show that {\agree} rewriting allows to simulate rewriting with
polarized cloning on graphs, which is defined in \cite{DuvalEP12} by
using the polarized graphs of Definition~\ref{def:pol-cat}. 
\hide{
Briefly,
in a polarized graph each node may be annotated with $+$ or $-$ (non
exclusively), and only nodes annotated with $+$ (resp. $-$) may have
outgoing (resp. incoming) edges.}
Polarization is used in rewriting to control the copies of edges not
matched but incident to the matched nodes. 
\hide{
that
connect a cloned node with nodes outside of the match: a node $n$ can
be copied with all its edges, only with its outgoing edges, only with
its incoming edges, or without any of its edges.}

\hide{Thus, polarization is a way to tune \sqpo rewriting:
indeed, \sqpo rewriting is the 
special case of rewriting with polarized cloning
where every node is copied with all its  edges. }
\hide{%
Rewriting with polarized cloning is defined in \cite{DuvalEP12}  
by introducing the notion of \emph{polarized graph}.
A  polarized graph is a graph where each node may be annotated 
with $+$ or $-$ (non exclusively).
In a polarized graph, only 
nodes annotated with $+$ (resp. $-$) may have outgoing 
(resp. incoming) edges. 
We recall the definitions of polarized graphs and 
of rewriting with polarized cloning in Definitions~\ref{defi:pol-cat} 
and~\ref{defi:psqpo}; 
more details on this subject can be found in \cite{DuvalEP12}. 
}

\hide{
In order to define rewriting with polarized cloning we need to introduce some 
functors relating categories $\Grpol$ and $\Gr$.}

\hide{*** Then, using the fact that the category 
of polarized graphs has a \emph{relative} partial map classifier,
we prove in Theorem~\ref{theorem:psqpo} that 
each rewriting step with polarized cloning
can be simulated by an {\agree} rewriting step in the category of graphs.
}

\begin{fact}
\label{defi:pol-depol}
The \emph{underlying graph} of a polarized graph
$\grpol{X}=(X,N_X^+,N_X^-)$ is $X$. This defines a functor
$\depol:\Grpol\to\Gr$ which has both a right- and a left-adjoint
functor denoted $\polar$ and $\polar^{\pm}:\Gr\to\Grpol$, resp.,
i.e.~$\polar^{\pm}\dashv \depol \dashv \polar$.

Functor $\polar$ maps each graph $X$ to the polarized graph
\emph{induced by} $X$, defined as $\grpol{X}=(X,N_X,N_X)$, and each
graph morphism $f: X \to Y$ to itself; it is easy to check that
$\polar(f): \polar(X) \to \polar(Y)$ is a \emph{strict} polarized
graph morphism. Furthermore we have that $\depol\circ\polar =
\Id_{\Gr}$, and we denote the unit of adjunction $\depol \dashv
\polar$ as $\unit:\Id_{\Grpol}\dotarrow \polar\circ\depol$, thus
$\unit_\grpol{X}: \grpol{X} \to \polar(\depol(\grpol{X}))$.

Functor $\polar^\pm$ maps each graph $X$ to the polarized graph
$\grpol{X}=(X,N^{+}_X,N^{-}_X)$, where a node is in $N_X^+$ (resp. in
$N_X^-$) if and only if it has at least one outgoing (resp. incoming)
edge in $X$.  Since $\depol$ has a left adjoint, we have that $\depol$
preserves limits and in particular pullbacks.
\end{fact}

The category $\Grpol$ has final pullback complements along strict monos:
their construction is given in \cite[Appendix]{DuvalEP12long}.
\hide{
Roughly speaking, a \psqpo rewrite step is made of 
a final pullback complement in $\Grpol$ followed by
a pushout in $\Gr$, thus generalizing a \sqpo rewrite step 
for which both phases are done in the same category.} 

\begin{definition}[{\psqpo} rewriting]
\label{defi:psqpo}
A \emph{\psqpo rewrite rule} $\rho$ 
is made of a span of graphs $\spa{L}{l}{K}{r}{R}$ and a 
polarized graph $\grpol{K}=(K,N_K^+,N_K^-)$ with underlying graph $K$. 
A \emph{\psqpo match} of the \psqpo rewrite rule $\rho$ is 
a mono $m:L \monoto G$ in $\Gr$.
A \emph{\psqpo rewriting step} $G \Rightarrow_{\rho,m}^{\psqpo} H$ 
is constructed as follows:
  \begin{enumerate}[label=(\alph*)]

  \item The left-hand-side $l$ of the rule $\rho$ gives rise to 
  a morphism $\wh{l}=\polar(l)\circ \unit_{\grpol{K}} : 
  \grpol{K}\to \polar(L)$ in $\Grpol$. 
  The match $m$ gives rise
  to a strict mono $\polar(m):\polar(L) \monoto \polar(G)$ in $\Grpol$. 
  Then  
  $\grpol{K} \uprto{n} \grpol{D} \uprto{g} \polar(G) $ is constructed as  the final pullback complement of 
  $\grpol{K} \uprto{\wh{l}} \polar(L) \uprto{\polar(m)} \polar(G)$ in category $\Grpol$.

  \item Since $\depol(\grpol{K})=K$, we get
  $\depol(n):K \to \depol(\grpol{D})$ in $\Gr$. 
  Then 
  $R \uprto{p} H \uplto{h} D $ 
  is built as the pushout of 
  $R \uplto{r} K \uprto{\depol(n)} \depol(\grpol{D})$ in  category $\Gr$.

  \end{enumerate}

\hide{
\vspace{-0.4cm}
$$ \xymatrix@C=3pc{
L \ar@{ >->}[d]_{m} & 
  \ar@{}[rd]|{FPBC} \polar(L) \ar@{ >->}[d]_{\polar(m)} & 
  \grpol{K} \ar[l]_{\wh{l}} \ar@{ >->}[d]^{n} & 
  \ar@{}[rd]|{PO} K \ar[r]^{r} \ar@{ >->}[d]_{\depol(n)} & 
  R \ar[d]^{p} \\
G & 
  \polar(G) & 
  \grpol{D}  \ar[l]^{g} & 
  \depol(\grpol{D}) \ar[r]_{h} & 
  H \\
}$$}
\end{definition}

Recall that, as observed in Sect.~\ref{sect:examples-classifiers}, category
$\Grpol$ has an \Str{}-partial map classifier $(\TT, \eta)$. This will
be exploited in the next result.
 
\begin{theorem}[\agree{} subsumes polarized node cloning on graphs]
\label{theorem:psqpo}
Let $\rho$ be a \psqpo rule made of span $L \uplto{l} K \uprto{r} R$
and polarized graph $\grpol{K}=(K,N_K^+,N_K^-)$. Consider the
component on $\grpol{K}$ of the natural transformation $\eta:
Id_{\Grpol} \dotarrow \TT$, and let $T_K=\depol(\TT(\grpol{K}))$ and
$t= \depol(\eta_{\grpol{K}}): \depol(\grpol{K}) \to
\depol(\TT(\grpol{K}))$, thus $t: K \to T_K$. Furthermore, let $m: L
\monoto G$ be a mono.  Then
$$G \To_{\rho,m}^{\psqpo} H \qquad \mbox{ if and only if } \qquad G \To_{(l,r,t),m}^{\agree{}} H$$
\hide{
\noindent In words, the application of rule $\rho$ to match $m$ using the \psqpo{} approach has exactly the same effect of applying to $m$ the span of $\rho$ enriched with the embedding $t: K \to T_K$ using the \agree{} approach in category $\Gr$.}
\end{theorem}

\begin{proof}
  The first phase of {\psqpo} rewriting consists of building the final
  pullback complement of $(\polar(m),\wh{l})$ in category
  $\Grpol$. According to Theorem~\ref{theorem:sqpo}, since $\polar(m)$
  is strict such final pullback complement can be obtained as the top
  square in the diagram below to the left, where both squares are
  pullbacks in $\Grpol$. The second phase consists of taking the
  pushout of morphisms $K \uprto{r} R$ and $\depol(n): K \to
  \depol(D)$ in $\Gr$.
 
  By applying functor $\depol$ to the left diagram we obtain the
  diagram below to the right in $\Gr$, where both squares are
  pullbacks because $\depol$ preserves limits. In fact, recall that
  $\depol \circ \polar = Id_\Gr$, that $K = \depol(\grpol{K})$ and
  that $t = \depol(\eta_{\grpol{K}})$; the fact that $T(L) =
  \depol(\TT(\polar(L)))$ can be checked easily by comparing the
  construction of the ($\Str$-)partial map classifiers in $\Gr$ and in
  $\Grpol$.
$$ \xymatrix@C=4pc@R=1.3pc{
\ar@{}[dr]|{PB} \polar(L) \ar@{ >->}[d]^{\polar(m)} 
  \ar@{ >->}@/_4ex/[dd]_(.6){\eta_{\polar(L)}}^(.6){=} & 
\grpol{K} \ar[l]_{\wh{l}} \ar@{ >->}[d]_{n} 
    \ar@{ >->}@/^4ex/[dd]^(.6){\eta_{\grpol{K}}}_(.6){=} \\
\ar@{}[dr]|{PB} \polar(G) \ar[d]^{\ol{\polar(m)}} & 
  \grpol{D} \ar[l]_{g} \ar[d]_{q=\oln} \\ 
  \TT(\polar(L)) & 
  \TT(\grpol{K}) \ar[l]^{\TT(\wh{l})} \\ } 
\qquad 
\xymatrix@C=4pc@R=1.3pc{
\ar@{}[dr]|{PB} L \ar@{ >->}[d]^{m} 
  \ar@{ >->}@/_4ex/[dd]_(.6){\eta_{L}}^(.6){=} & 
K \ar[l]_{l} \ar@{ >->}[d]|{\depol(n)} 
    \ar@{ >->}@/^6ex/[dd]^(.6){t}_(.6){=} \\
\ar@{}[dr]|{PB} G \ar[d]^{\olm} & 
  \depol(\grpol{D})  \ar[l] \ar[d] \\ 
  T(L) & 
  T_K \ar[l] \\ }
$$

Now, the first phase of \agree{} rewriting with rule $(l,r,t)$ and match $m$ consists of  taking the pullback in $\Gr$ of
 $\olm$ and the only arrow $T_k \to T(L)$ that makes the outer square of the right diagram a pullback. This arrow 
 is precisely $\depol(\TT(\wh{l}))$, and therefore the pullback is exactly the lower square of the right diagram. The second 
 phase consists of taking the pushout of  $K \uprto{r} R$ and of the only arrow $K \to \depol(D)$ that makes the diagram commute; 
 but $\depol(n)$ is such an arrow, thus the pushout is the same computed by the \psqpo{} approach and this concludes the proof.

\hide{to the category $\Grpol$ and to the strict monos. 
It means that in Definition~\ref{defi:psqpo} the final pullback complement 
$(n,g)$ of $(\polar(m),\wh{l})$ 
can be obtained from the pullback 
$(q,g)$ of $(\ol{\polar(m)},\TT(\wh{l}))$ 
by defining $n$ as the unique map such that 
$g\circ n = \polar(m)\circ \wh{l} $ and $q\circ n = \eta_{\grpol{K}}$:
see diagram below on the left. 
The functor $\depol$ is applied to this construction;  
since it preserves pullbacks, 
we get a similar construction in $\Gr$:
see diagram below on the right, where $D=\depol(\grpol{D})$.
Theorem~\ref{theorem:sqpo} proves that we get 
the final pullback complement of $(m,l)$ in $\Gr$. 
Since in the \psqpo approach the pushout phase is done in $\Gr$, 
both $\rho$ and $\rho^\pm$ produce the same result $H$. 
}

\end{proof}




\section{Related Work and Discussion}
\label{sec:discussion}

In this paper we presented the basic definitions of a new approach to
algebraic graph rewriting, called \agree{}. We showed that this
approach subsumes other algebraic approaches like SqPO
(Sesqui-pushout) with injective matches (and therefore DPO and SPO
under mild restrictions, see~\cite[Propositions~12 and~14]{CorradiniHHK06}), as well as its
polarised version PSqPO.  The main feature provided by this approach
is the possibility, in a rule, of specifying which edges shall be
copied as a side effect of the copy of a node. This feature offers
new facilities to
specify applications in which copy of nodes shall be done in an
unknown context, and thus it is not possible to describe in the
left-hand side of the rule all edges that shall be copied together
with the node. As an example, the anonymization of parts of a social
network was described in Sect.~\ref{sec:examples}.

The idea of controlling explicitly in the rule how the right-hand side
should be embedded in the context graph is not new in graph rewriting,
as it is a standard ingredient of the algorithmic approaches.  For
example, in Node Label Controlled (NLC) graph rewriting and its
variations~\cite{EngelfrietR97} productions are equipped with
\emph{embedding rules}, which allow one to specify how the right-hand
side of a production has to be embedded in the context graph obtained
by deleting the corresponding left-hand side. The name of our approach
is reminiscent of those older ones.

\comment{Among the approaches based on categorical constructions, besides those
already discussed, one may wonder whether an \agree rewriting step can
be simulated by one parallel transformation following the
\emph{amalgamation} concept (e.g., \cite{BFH87}).  Roughly speaking,
the notion of amalgamation of graph transformation produces the effect
of several rules in one shot. This is a kind of parallel graph
transformation where the effects of the involved basic rules are
added. The theory of amalgamation has been investigated for the main
graph transformation approaches such as SPO, DPO and more recently
SqPO \cite{Lowe15}. As said earlier, \agree rules generalize classical
algebraic approaches and thus the use of parallel reductions
cannot help to reach all the capabilities offered by the embedding parts
in the \agree rewriting steps.  However, the notion of amalgamation of
the \agree rules still remains as a subject of future investigation.}
Adaptive star grammars \cite{DHJM10} is another framework where node
cloning is performed by means of rewrite rules of the form $ S ::= R $
where graph $S$ has a shape of a star and $R$ is a graph. Cloning
operation, see \cite[Definitions~5 and~6]{DHJM10}, shares the same
restrictions as the sesqui-pushout approach: nodes are
cloned with all their incident edges.

In \cite{Lowe10} a general framework for graph transformations in
span-categories, called \emph{contextual graph rewriting}, briefly CR, has
been proposed.  Using CR, thanks to the notions of rule and of match
that are more elaborated than in other approaches, it is possible to specify 
cloning as in AGREE rewriting, and even more general transformations: e.g., one may create
multiple copies of nodes/edges as a side effect, not only when
cloning items. 
The
left-hand sides of CR rules allow to specify elements that must exist
for the rule to be applicable, called $E$, and also a
context for $E$, i.e.~a part of the graph that will be universally
quantified when the rule is applied, called $U$. A third
component plays the role of embedding the context $U$ in the rest
of the graph. The rule for copying a web page shown in
Figure~\ref{fig_webpagecopy}  could be specified using  CR 
as rule
$E \rightarrowtail U \rightarrowtail L \leftarrow K \rightarrow R$,
where $E = L1, U = L = T(L1)$ and $ K = R = TK1$. Finding a match for
a rule in a graph $G$ involves
finding a smallest subgraph of $G$ that contains $E$ and its complete
context. Thus, even if CR is more general, our approach enhances the expressiveness
of classical algebraic approaches with a form of controlled cloning using
simpler and possibly more natural rules.

%
%
%
%
%
%
 
Bauderon's pullback approach \cite{BauderonJ01} is also related to our
proposal. It was proposed as an algebraic variant of the above
mentioned NLC and ed-NLC algorithmic approaches. Bauderon's approach
is similar, in part, to the pullback construction used in our first
phase of a rewriting step, but a closer analysis is needed and is
planned as future work.  We also intend to explore if there are
relevant applications where \agree{} rewriting in its full generality
(i.e., with possibly non-local rules) could be useful.


Concerning the applicability of our approach to other structures, in practice the requirement of existence of partial maps classifiers looks quite demanding. \agree{} rewriting works in categories of typed/colored graphs, which are used in several applications, because they are slice categories over graphs, and thus toposes. But even more used are the categories of attributed graphs  \cite{DuvalEPR14}, which are not toposes. Under which conditions our approach can be extended or adapted to such structures is an interesting topic that we intend to investigate.

%
%


\subsection*{Acknowledgments}
We are grateful to the anonymous reviewers of former versions of this paper for the insightful and constructive criticisms.

\bibliographystyle{splncs03}


\appendix
\section{Proofs}
\label{sec:AppProofs}

This section is devoted to the proof of Proposition~\ref{prop:locality} and to part of the proof of 
Theorem~\ref{theorem:sqpo}. Let $\catC$ be a category satisfying all conditions of Definition~\ref{def:complement}, where  $(T, \eta)$ is an  \M-partial map classifier.
Let us start with a technical lemma.

\begin{lemma}
\label{lemma:locality}
Object $T(L)\setminus L $ is isomorphic to $T(\ini)$ for each $L$, and furthermore
 $T(l)\setminus l: T(K)\sm K \to T(L)\sm L$ is an iso for each $l:K\to L$.
\end{lemma}

\noindent
\begin{minipage}{.65\textwidth}
\emph{Proof.}\quad First, let us look at the diagram to the right
where $L$ is any object. 
In this diagram the top square is a pullback of shape (\ref{pb:eta}) 
and the bottom square is a pullback because,
up to the isomorphism between $\fin$ and $T(\ini)$ we may replace 
$\fin_{T(\ini)}:T(\ini)\rto \fin$ by $T(\id_{\ini}):T(\ini)\rto T(\ini)$
and $\ffalse:\fin\to T(\fin)$ by $T(\inifin):T(\ini)\rto T(\fin)$, 
so that the bottom square  becomes the image by $T$ of a pullback square. 
Thus, $T(L)\setminus L $ is isomorphic to $T(\ini)$ 
and, up to this iso, $T(L)\setminus \eta_L$ is $T(\init_L):T(\ini) \to T(L)$. 

\medskip
Now, let us look at the diagram to the right
where $l:K\to L$ is any arrow. 
In this diagram the top square is a pullback of shape (\ref{pb:eta}) 
and the bottom square is a pullback because it is the image by $T$ of a pullback square. 
Thus, $T(l)\setminus l:T(K)\setminus K \to T(L)\setminus L$ is an iso. 
\end{minipage}
\begin{minipage}{.35\textwidth}
\centerline{$ \xymatrix@C=4.0pc@R=1.8pc{
\ar@{}[rd]|{PB} 
L \ar@{ >->}[d]|{\bx{\eta_L}} \ar[r]|{\bx{\final_L}} &
\fin \ar@{ >->}[d]|{\bx{\ttrue}}  \\  
\ar@{}[rd]|{PB} 
T(L) \ar[r]|{\bx{T(\final_L)}} & 
T(\fin) \\  
T(\ini) \ar@{ >->}[u]|{\bx{T(\init_L)}} \ar[r]|(.6){\bx{\fin_{T(\ini)}}} &
\fin \ar@{ >->}[u]|{\bx{\ffalse}} \\
} $ }
\medskip
\centerline{\xymatrix@C=4.0pc@R=1.8pc{
\ar@{}[dr]|{PB} 
  K \ar@{ >->}[d]|{\bx{\eta_K}} \ar[r]|{\bx{l}} &
  L \ar@{ >->}[d]|{\bx{\eta_L}} \\ 
\ar@{}[dr]|{PB} 
  T(K) \ar[r]|{\bx{T(l)}} &
  T(L) \\
  T(\ini) \ar@{ >->}[u]|{\bx{T(\init_K)}} \ar[r]|{\bx{\id_{T(\ini)}}} & 
  T(\ini) \ar@{ >->}[u]|{\bx{T(\init_L)}} \\
  } }
\end{minipage}

\begin{proof}[of Proposition~\ref{prop:locality}]
Let us recall the statement of the proposition, for the readers' convenience:

\begin{quote} 
Let $\rho = (l,r,t)$ be a local rule.
Then, with the notations as in diagram~(\ref{eq:agree-rew}), for each match $L \upmonoto{m} G$ the resulting rewrite step $G \To_{\rho,m} H$ is local.
\end{quote}

By Definition~\ref{def:embSpec} we have to show that if $\ol{t}:T_K\to T(K)$ 
is such that $\ol{t}\setminus \id_K:T_K\setminus K \to T(K)\setminus K$ is an iso, i.e.~the rule is local, then
arrow  $g\setminus l: D\setminus K \rto G\setminus L$ is an iso as well.
Consider the diagram in Figure~\ref{fig:sm}, where the left part depicts the first phase of an {\agree} rewriting step, together with several arrows to the \M-subobject classifier $T(1)$. The right part is obtained by pulling back 
(part of) the left part along $\mathit{false}:1 \to T(1)$, obtaining the depicted strict complements (see Definition~\ref{def:embSpec}). Now, in triangle $(\ddagger)$ arrow $\ol{t}\setminus \id_K$ is iso by hypothesis, and $T(l)\sm l$ is iso by Lemma~\ref{lemma:locality}. Therefore $l' \sm l$ is an iso as well. Furthermore the square
around $1$ is a pullback, because it is obtained by pulling back (along $\mathit{false}$) the pullback around $T(1)$, and therefore $g\sm l$ is an iso.

\begin{figure}[htb]
$$ \xymatrix@C=2pc@R=2pc{
L \ar@{ >->}[dd]|{\bx{m}}   \ar@/_3ex/[dddd]|{\bx{\eta_L}}  
&
& K \ar[ll]|{\bx{l}}  \ar@{ >->}[dd]|{\bx{n}} \ar@{ >->}@/_3ex/[dddd]|(.2){\bx{t}} \ar@{ >->}@/^4ex/[dddddd]|(.49){\bx{\eta_K}}
\\
& & & {G\sm L} \ar[dd]|(.33){\hole}|{\bx{\olm\sm id_L}}|(.65){\hole} \ar@{-->}[llld]|(.3)\hole|(.37)\hole|(.45)\hole|(.6){\bx{G\sm m}}  \ar@{..>}[rd]|\hole
& & {D\sm K} \ar[dd]|{\bx{n'\sm id_K}} \ar@{-->}[llld]|(.55){\bx{D\sm n}}  \ar@{..>}[ld]  \ar[ll]|{\bx{g \sm l}}
\\
G  \ar@{..>}[rd]|{\bx{\chi_m}} \ar[dd]|{\bx{\olm}}
& & D \ar[dd]|{\bx{n'}}  \ar@{..>}[ld]|{\bx{\chi_n}} \ar[ll]|(.13)\hole|{\bx{g}}
 & &  1  \ar@{=>}[llld]|(.45){\bx{\mathit{false}}}|(.59)\hole|(.68)\hole|(.75)\hole \\ 
& T(1) 
&  & {T(L) \sm L}  \ar@{..>}[ru] \ar@{-->}[llld]|(.26)\hole|(.37)\hole|(.42)\hole|(.53)\hole|(.62)\hole 
 & &   {T_K \sm K}   \ar[dd]|{\bx{\olt \sm id_K}} \ar@{-->}[llld]|(.7){\bx{T_K\sm t}}  \ar@{..>}[lu]  \ar[ll]|{\bx{l' \sm l}}\\ 
T(L) \ar@{..>}[ru]|{\bx{T(1_L)}}
& & T_K  \ar[ll]|{\bx{l' = \varphi(t,l)}}  \ar@{..>}[ul]|(.4){\bx{\chi_t}} \ar[dd]|{\bx{\olt}}
\\
& & & \ar@{}[uurr]|(.67){(\ddagger)}  & &  T(K) \sm K  \ar[uull]|{\bx{T(l)\sm l}}|(.77)\hole  \ar@{-->}[llld]|{\bx{T(K)\sm \eta_K}}  \ar@{..>}[uuul]|(.61)\hole|(.68)\hole
\\ 
& & T(K) \ar@{..>}[uuul]|{\bx{T(1_K)}} \ar[uull]|{\bx{T(l)}}  
}$$
\caption{Transformation of strict complements in {\agree}}
\label{fig:sm}
\end{figure}
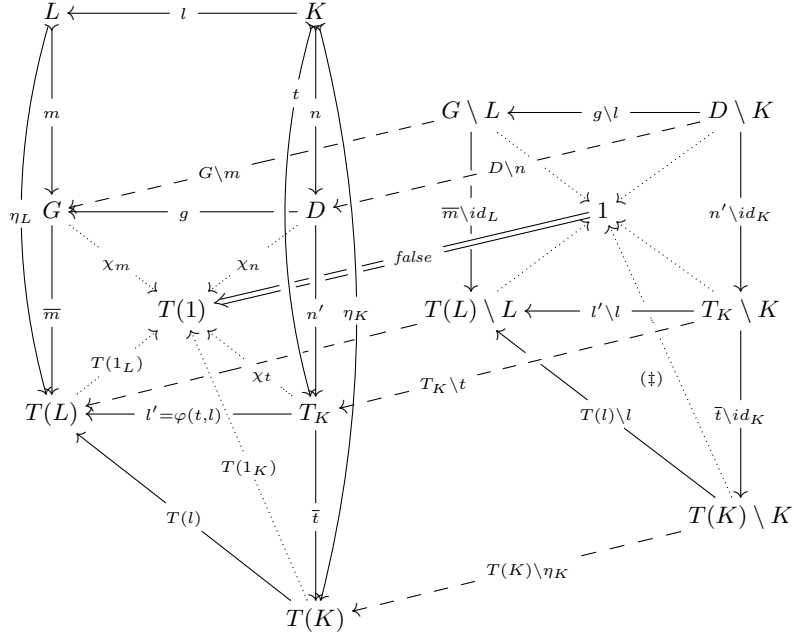
\end{proof}

\pagebreak
\begin{proof}[Uniqueness part of the proof of Theorem~\ref{theorem:psqpo}]
Let as redraw the right diagram of Figure~\ref{fig:FPBCasPB} for the reader's convenience, 
enriched with some additional information.
$$ \xymatrix@C=1pc@R=0.8pc{
& & & X \ar@{..>}@/_3ex/[ddl]|{\bx{v}} \ar@{..>}@/^3ex/[dddr]|{\bx{w}} \ar@{..>}[d]_z \ar@<0ex>@{}[ddl]|(.4){\ideq{6}} \ar@<0ex>@{}[dddr]|(.4){\ideq{7}}\\
& & & K'   \ar@/_2ex/[dlll]|{\bx{d}}   \ar@<1ex>@{}[dlll]|{\ideq{1}}    \ar[ddr]|{\bx{e}}  \ar[dl]|{\bx{h}}   
 \ar@<0ex>@{}[ddd]|{\ideq{2}} \\
L \ar@{ >->}[ddr]|{\bx{m}}     \ar@{ >->}[dddd]|{\bx{\eta_L}} 
& & K \ar[ll]|{\bx{l}}    \ar@{ >->}[dddd]|(.23){\bx{\eta_K}}|(.4){\hole}|(.53){\hole}    \ar@{ >->}[ddr]|(.35){\bx{n}}|(.63){\hole} \\
& & & & D'   \ar@{..>}[dl]|{{g}}   \ar@/_1ex/[dlll]|{\bx{f}}   \ar@/^3ex/[dddll]|{\phantom{\big(}\varphi(e,h)\phantom{\big(}}    \ar@/^13ex/[dddllll]|{\phantom{\big(}\varphi(e,d)\phantom{\big(}}    
\ar@<1ex>@{}[dlll]|{\ideq{3}}   \ar@<1ex>@{}[dddll]|(.4){\ideq{5}}\\
& G  \ar[ddl]|{\bx{\overline{m}}}  \ar@{}[ddr]|{\ideq{4}}   
& & D  \ar[ddl]|(.35){\phantom{\big(}n'\phantom{\big(}}  \ar[ll]|(.3){\bx{a}}\\
& \\
T(L)  
& & T(K) \ar[ll]|{\bx{T(l)}} 
}$$
We have to prove that the arrow $D' \uprto{{g}} D$, that was shown to exists in the first part of the proof, is the only arrow that satisfies  $n \circ h = {g} \circ e$  and $a \circ {g} = f$. 
Suppose indeed that  $D' \uprto{\hat{g}} D$ is another arrow such that  $n \circ h \eqid{2} \hat{g} \circ e$  and $a \circ \hat{g} \eqid{3} f$. Since \ideq{4} is a pullback, in order to show that $\hat{g} = g$ it is sufficient 
to show that $n' \circ \hat{g} \eqid{5} \varphi(e,h)$, because commutativity of $\ideq{3}$ and $\ideq{5}$ uniquely determines a mediating arrow $D' \to D$. To show $n' \circ \hat{g} \eqid{5} \varphi(e,h)$, recall that 
by the properties of the \M{}-partial map classifier $(T,\eta)$ there is a unique arrow $D' \uprto{\varphi(e,h)} T(K)$ such that $\eta_K \circ h = \varphi(e,h) \circ e$ and the square is a pullback. Therefore it is sufficient to show that $K' \uprto{h} K \upmonoto{\eta_K} T(K) \uplto{n' \circ \hat{g}} D' \uplto{e} K'$ is a pullback. 

First, it commutes, as $n' \circ \hat{g} \circ e \eqid{2} n' \circ n\circ h = \eta_K \circ h$. Next, let $\langle X, X \uprto{v} K, X \uprto{w} D'\rangle$ be such that $\eta_K \circ v = n' \circ \hat{g} \circ w$.  We have to show that there is a unique $X \uprto{z} K'$ such that $v \eqid{6} h \circ z$ and $w \eqid{7}  e \circ z$.  For \emph{existence}, an arrow  $X \uprto{z} K'$ is determined by exploiting the pullback   $\eta_L \circ d = \varphi(e,d) \circ e$ (it is a pullback again by the properties of $(T,\eta)$).

In fact we have $\eta_L \circ (l \circ v) = T(l) \circ \eta_K \circ v = T(l) \circ \varphi(e,h) \circ w = \varphi(e,d) \circ w$. Thus there is an arrow  $X \uprto{z} K'$ such that both \ideq{7} and $l \circ v \eqid{8} d \circ z$ hold. It remains to show \ideq{6}, i.e.\ that $h\circ z = v$. By exploiting pullback  $\eta_L \circ l = T(l) \circ \eta_K$, it is sufficient to show that (i) $l\circ h\circ z = l \circ v$ and (ii) $\eta_K \circ h\circ z = \eta_K \circ v$. In fact, we have (i) $l \circ h\circ z \eqid{1} d \circ z \eqid{8}  l \circ v$, and (ii) $\eta_K \circ h\circ z =  n' \circ n  \circ h \circ z   \eqid{2} n' \circ \hat{g}  \circ e \circ z   \eqid{7}   n' \circ \hat{g}  \circ w = 
\eta_K \circ v$.  Finally, the \emph{uniqueness} of $X \uprto{z} K'$ follows by the observation that commutativity of  \ideq{6} and \ideq{7} uniquely determines a mediating morphism to $K'$ regarded as pullback object of $K' \uprto{e} D' \uprto{\varphi(e,h)} T(K) \uplto{\eta_K} K \uplto{h} K'$. 

\end{proof}
\section{The final pullback complement theorem, revisited}
\label{app:abstract}

This Appendix is dedicated to a more abstract, equivalent presentation of the 
statement of 
Theorem~\ref{theorem:sqpo} and of its proof. 
By exploiting the characterization of the final pullback complement 
as an adjoint functor, we get a proof which hides 
some diagram chasing by using general properties of 
partial map classifiers and adjunctions. 
First we state a lemma about decomposing the arrow $\varphi(m,f)$,
then we recall the definitions of slice categories and pullback functors,
and finally we get a new point of view on Theorem~\ref{theorem:sqpo}. 

\begin{lemma}[Decompositions of $\varphi(m,f)$] 
\label{lemma:mf}
Let $\catC$ be a category with pullbacks 
and with an $\M$-partial map classifier $(T,\eta)$ 
for a stable system of monos $\M$. 
For each \M-partial map $(m,f):Z\parto Y$, with $m:X\monoto Z$, 
we have $T(f)\circ \olm = \varphi(m,f) $.
If in addition $(m,f)$ is the pullback of some $(n,g)$ 
with $n:Y\monoto W$ in \M{}, then 
$T(f)\circ \olm = \varphi(m,f) = \oln \circ g$.
\end{lemma}

\begin{proof}
For the first point, the left diagram below is composed of two pullbacks 
of shape~(\ref{pb:olm}) and~(\ref{pb:eta}), respectively, 
therefore it is a pullback. Since it has shape~(\ref{pb:pmc}), we conclude
that $ T(f)\circ \olm  = \varphi(m,f)$. 

$$\xymatrix@C=3pc{
\ar@{}[rd]|{PB~(\ref{pb:olm})} X \ar@{ >->}[d]|{\bx{m}} \ar[r]|{\bx{\id_X}} & 
  \ar@{}[rd]|{PB~(\ref{pb:eta})} X \ar@{ >->}[d]|{\bx{\eta_X}} \ar[r]|{\bx{f}} & 
  Y \ar@{ >->}[d]|{\bx{\eta_Y}} \\ 
Z \ar[r]|{\bx{\olm}}  \ar@{-->}@/_3ex/[rr]|{\bx{\varphi(m,f)}}   & 
  T(X) \ar[r]|{\bx{T(f)}} & 
  T(Y) } \qquad
\xymatrix@C=3pc{
\ar@{}[rd]|{PB} X \ar@{ >->}[d]|{\bx{m}} \ar[r]|{\bx{f}} & 
  \ar@{}[rd]|{PB~(\ref{pb:olm})} Y \ar@{ >->}[d]|{\bx{n}} \ar[r]|{\bx{\id_X}} & 
  Y \ar@{ >->}[d]|{\bx{\eta_Y}} \\ 
Z \ar[r]|{\bx{g}} \ar@{-->}@/_3ex/[rr]|{\bx{\varphi(m,f)}} & 
  W \ar[r]|{\bx{\oln}} & 
  T(Y) } 
$$ 
For the second point, similarly, the right diagram above is the composition 
of a pullback of shape~(\ref{pb:olm}) and of the left square that is pullback 
by assumption, 
thus it is a pullback. Since  it has shape~(\ref{pb:pmc}), we can  conclude
that $ \oln \circ g =  \varphi(m,f)$.
\end{proof} 

For each object $X$ in a category $\catC$,
the \emph{slice category} over $X$ is denoted $\catC\slice X$:
its objects are the arrows $f:Y\to X$ in $\catC$
and an arrow $g:f_1\to f_2$ in $\catC\slice X$, 
with $f_1:Y_1\to X$ and $f_2:Y_2\to X$ in $\catC$, 
is an arrow $g:Y_1\to Y_2$ in $\catC$ such that $f_2\circ g= f_1$. 
For each endofunctor $F:\catC\to\catC$ and each object $X$ in $\catC$, 
let us still denote by $F$ the functor $F:\catC\slice X \to \catC\slice F(X)$
which maps each object $f$ of $\catC\slice X$ to $F(f)$ 
and each arrow $g:f_1\to f_2$ of $\catC\slice X$ to $F(g)$. 

For each arrow $m:X\to Z$ in a category $\catC$ with pullbacks, 
the \emph{pullback functor} associated with $m$ 
is denoted $\pb_m:\catC\slice Z \to \catC\slice X$; 
on objects, it maps each $h$ to $f=\pb_m(h)$ 
such that the square below on the left is a pullback square;
on arrows, using the decomposition property of pullbacks,
it maps each $k:h_1\to h_2$ to the unique $g=\pb_m(k):f_1\to f_2$,
where $f_1=\pb_m(h_1)$ and $f_2=\pb_m(h_2)$, 
such that $(g,n_1)$ is a pullback of $(k,n_2)$ (below on the right).
In fact, ``the'' pullback functor is defined only up to isomorphism, 
but this will not raise any problem. 
$$
\xymatrix@C=6pc@R=3pc{ 
\ar@{}[rd]|{PB}
Y \ar[r]|{\bx{f}} \ar[d]|{\bx{n}} & X \ar[d]|{\bx{m}} \\
W \ar[r]|{\bx{h}} & Z \\ 
} \qquad  \qquad 
\xymatrix@C=1.5pc@R=1pc{ 
Y_1 \ar[rrr]|(.5){\bx{f_1}} \ar[dd]|{\bx{n_1}} \ar[rd]|{\bx{g}} &&& X \ar[dd]|(.5){\bx{m}} \\
& Y_2 \ar[rru]|{\bx{f_2}} \ar[dd]|(.35){\bx{n_2}} &&  \\
W_1 \ar[rrr]|(.4)\hole|(.7){\bx{h_1}} \ar[rd]|{\bx{k}} &&& Z \\ 
& W_2 \ar[rru]|{\bx{h_2}} &&  \\ 
} 
$$ 

The next result rephrases part of Theorem 4.4 of~\cite{DT87} (see also~\cite{CorradiniHHK06}).
\begin{theorem}[final pullback complements as right adjoints]
\label{thm:DT}
Let $\catC$ be a category with pullbacks and $m:X\to Z$ be an arrow of $\catC$. Then the following are equivalent.
\begin{enumerate}
\item  The pullback functor $\pb_m:\catC\slice Z \to \catC\slice X$
has a right adjoint $G$
and the counit of the adjunction is a natural isomorphism.

\item  Arrow $m$ \emph{has final pullback complements}, i.e., for each $f: Y \to X$ there is a pair of composable 
arrows $Y \uprto{n} W \uprto{h} Z$ which are a final  pullback 
complement of $Y \uprto{f} X \uprto{m} Z$.
\end{enumerate}
In addition, if the previous points hold then  $h$ and $G(f)$ coincide up to isomorphism.
\end{theorem}

\hide{  
The coherence of Definition~\ref{def:fpbc} with the following more abstract 
Definition~\ref{def:fpbc-abstract} is explained in \cite{DT87,CorradiniHHK06}. 

%
%

\begin{definition}[final pullback complement functor]
\label{def:fpbc-abstract}
Let $\catC$ be a category with pullbacks.
For each $m:X\to Z$ in $\catC$, if the pullback functor $\pb_m$
has a right adjoint $G$
and if the counit of this adjunction is a natural isomorphism,
then $G$ is called
the \emph{final pullback complement functor} associated with $m$.
\end{definition}

Then by \cite[Theorem 4.4]{DT87}
there is a final pullback complement functor associated with $m$
if and only if
for every $f:Y\to X$ the pair $(f,m)$ has a final pullback complement
in the sense of Definition~\ref{def:fpbc}
}


Let $\catC$ be a category with pullbacks 
and let $\M$ be a stable system of monos of $\catC$. 
Then the composition of consecutive $\M$-partial maps is defined 
in the usual way, using a pullback in $\catC$. 
This yields the category $\catP$ of $\M$-partial maps over $\catC$   
and the inclusion functor $I:\catC\to\catP$,  
which maps each object $X$ to $X$ and 
each arrow $f:X\to Y$ to $(\id_X,f):X\parto Y$. 
According to \cite[Sec.2.1]{CL2}, $\catC$ has an $\M$-partial map classifier
if and only if the functor $I$ has a right adjoint $E:\catP\to\catC$, 
and then the $\M$-partial map classifier $(T,\eta)$ is made of the endofunctor 
$T=E\circ I$ on $\catC$ and of the unit of 
the adjunction, $\eta: Id_{\catC} \dotarrow T$.
Thus, functor $T$  is defined as $T(X)=E\circ I(X) = E(X)$ for each object $X$ and 
$T(f)=E(\id_X,f)$ for each arrow $f:X\to Y$. 
Now, exploiting Theorem~\ref{thm:DT} we can state and prove Theorem~\ref{theorem:sqpo} 
 in a more abstract framework, as follows. 
\begin{theorem}[building final pullback complements (revisited)]
\label{theorem:sqpo-abstract}
Let $\catC$ be a category with pullbacks 
and with an $\M$-partial map classifier $(T,\eta)$ 
for a stable system of monos $\M$. 
Then for each mono $m:L\monoto G$ in $\M$ the functor  
$\fpbc_m = \pb_{\olm}\circ T :\catC\slice L \to \catC\slice G$ is the right 
adjoint to functor $\pb_m :\catC\slice G \to \catC\slice L$.
In addition, the counit of the adjunction is a natural 
isomorphism. 
\end{theorem}

\begin{proof}
Let us sketch this proof by describing the unit $\unit$ and counit $\counit$
of the adjunction $\pb_m \dashv \fpbc_m$.
For the counit, since $\olm\circ m = \eta_L$ we have 
$\pb_m\circ\pb_{\olm} = \pb_{\eta_L} $,
and since the natural transformation $\eta$ is cartesian 
we have $\pb_{\eta_L}\circ T \cong \Id_{\catC\slice L} $. 
Then the counit $ \counit:  \pb_m \circ \fpbc_m \To \Id_{\catC\slice L} $
is the resulting natural isomorphism. 
For the unit, let $g:D\to G$ be an object in $\catC\slice G$ and  
let $l=\pb_m(g): K\to L$ in $\catC\slice L$ (see the diagrams below). 
Let $n:K\monoto D $ be the fourth arrow in this pullback,
then $n\in\M$ by stability and 
by Lemma~\ref{lemma:mf} we have $T(l)\circ \oln = \olm \circ g$. 
Let $g'= \pb_{\olm}\circ T(l):D'\to G $ and 
let $q:D'\to T(K)$ be the fourth arrow in this pullback. 
By definition of pullback, there is a unique arrow $\unit_g:D\to D'$ 
such that $g'\circ \unit_g=g $ and $ q\circ \unit_g=\oln$. 
It follows that $\unit_g:g\to g'$ is an arrow in $\catC\slice G$. 
Moreover, let $n'=\unit_g\circ n: K\to D'$,
then $q\circ n' = q\circ \unit_g\circ n = \oln\circ n = \eta_K$.
Since $\eta$ is cartesian, 
the decomposition property of pullbacks implies that $(n',l)$ 
is the pullback of $(m,g')$, so that $n'$ is in $\M$ and $q=\ol{n'}$.
Then it can be checked that the $\unit_g$ arrows defines 
a natural transformation
$\unit: \Id_{\catC\slice G} \To  \fpbc_m \circ \pb_m$,
which is the unit of the adjunction. 
%
$$ \xymatrix@C=4pc{
\ar@{}[rd]|{PB} L \ar@{ >->}[d]_{m} & 
  K \ar[l]|{\bx{l}} \ar@{ >->}[d]^{n} \\
\ar@{}[rd]|{=} G \ar[d]_{\olm} & 
  D \ar[l]|{\bx{g}} \ar[d]^{\oln} \\ 
T(L) & 
  T(K) \ar[l]|{\bx{T(l)}} \\ 
} \qquad \qquad 
\xymatrix@C=4pc{
\ar@{}[rd]|{PB} L \ar@{ >->}[d]_{m} & 
  K \ar[l]|{\bx{l}} \ar@{ >->}[d]^{n'} \\
\ar@{}[rd]|{PB} G \ar[d]_{\olm} & 
  D' \ar[l]|{\bx{g'}} \ar[d]^{q=\ol{n'}} \\ 
T(L) & 
  T(K) \ar[l]|{\bx{T(l)}} \\ 
} $$
\end{proof}


\end{document}